\RequirePackage{snapshot}
\documentclass[11pt,a4paper]{article}
\usepackage[utf8]{inputenc}
\usepackage{amsmath,amsfonts,amssymb,amsthm}
\usepackage{fullpage}
\usepackage{color}
\usepackage{xspace}
\usepackage{hyperref}
\usepackage{bbm}
\usepackage{empheq}
\usepackage{authblk}
\usepackage{xargs}[2008/03/08]

\title{\textbf{Stabilizing Consensus with Many Opinions}}

\author[1]{L. Becchetti}
\author[2]{A. Clementi}
\author[1]{E. Natale}
\author[2]{F. Pasquale}
\author[3]{L. Trevisan}

\affil[1]{\emph{Sapienza} Universit\`a di Roma, {\tt becchett@dis.uniroma1.it}, {\tt natale@di.uniroma1.it}}
\affil[2]{Universit\`a \emph{Tor Vergata} di Roma, {\tt clementi@mat.uniroma2.it}, {\tt pasquale@mat.uniroma2.it}}
\affil[3]{U.C. Berkeley, {\tt luca@berkeley.edu}}

\makeatletter
\newtheorem*{rep@theorem}{\rep@title}
\newcommand{\newreptheorem}[2]{%
\newenvironment{rep#1}[1]{%
 \def\rep@title{#2 \ref{##1}}%
 \begin{rep@theorem}}%
 {\end{rep@theorem}}}
\makeatother

\newtheorem{definition}{Definition}[section]
\newtheorem{corollary}[definition]{Corollary}
\newtheorem{lemma}[definition]{Lemma}
\newtheorem{theorem}[definition]{Theorem}
\newreptheorem{theorem}{Theorem}

\newcommand{\Prob}[2]{\mathbf{P}_{#1} \left( #2 \right)}
\newcommand{\expec}[2]{\mathbf{E}\left[\left. #1 \;\right| #2 \right]}
\newcommand{\Expec}[2]{\mathbf{E}_{#1} \left[ #2 \right]}

\newcommand{\poly}{ {\mathrm{poly}}}
\newcommand{\polylog}{ {\mathrm{polylog} } }
\newcommand{\bigO}{\mathcal{O}}

\newcommand{\skproof}{\noindent\textit{Sketch of Proof. }}
\newcommand{\ideaproof}{\noindent\textit{Idea of the proof. }}

\newcommand{\imax}{ {\mathbbm{M}}}
\newcommand{\jmax}{ {\mathbbm{2M}}}
\newcommand{\col}{\ensuremath \mathbf{C}}

\newcommand{\goodcols}{\mathcal{C}}
\newcommand{\badcols}{\bar{\mathcal{C}}}
\global\long\def\bad{S}
\newcommand{\bigc}{B}

\newcommand\advconst{\beta}
\newcommand{\adv}{\advconst \sqrt{n}/(k^{\frac 52} \log n)}
\newcommand{\fadv}{ \frac{ \advconst \sqrt{n}}{ k^{\frac 52} \log n }}
\newcommand{\advletter}{D}
\newcommandx\advi[1][usedefault, addprefix=\global, 1=]{\advletter_{i}^{#1}}

\newcommand{\smallbiasconst}{\alpha}
\newcommand{\smallsizeconst}{\gamma}
\newcommand{\smallsizedenom}{k^{\frac 32}}
\newcommand{\smallsize}{ \smallsizeconst \sqrt{ n }/\smallsizedenom}
\newcommand{\fsmallsize}{ \smallsizeconst \frac{\sqrt{  n}}{\smallsizedenom}}
\newcommand{\smallbias}{ \smallbiasconst \sqrt{ n / k }}

\newcommand{\dieadvfactor}{ \psi }
\newcommand{\border}{ m }
\newcommand{\hhyp}{condition $\mathcal{H}$\xspace}

\global\long\def\constone{\alpha_{1}}
\global\long\def\consttwo{\alpha_{2}}
\global\long\def\constthree{\alpha_{3}}

\newcommandx\mincol[1][usedefault, addprefix=\global, 1=]{\mathbbm m\left(#1\right)}

\newcommand{\timestart}{t'} 
\newcommand{\timeend}{t''} 

\renewcommand{\leq}{\leqslant}
\renewcommand{\geq}{\geqslant}
\renewcommand{\le}{\leqslant}
\renewcommand{\ge}{\geqslant}

\begin{document}

\maketitle

\begin{abstract}
We consider the following distributed consensus problem: Each node in a
complete communication network of size $n$  initially holds an  \emph{opinion},
which is chosen arbitrarily from a finite set $\Sigma$. The
system must converge toward a consensus state in which all, or almost all
nodes, hold the same opinion. Moreover, this opinion should be \emph{valid},
i.e., it should be one among those initially present in the system. This
condition should be met even in the presence of an adaptive, malicious
adversary who can modify the opinions of a bounded number of nodes in every
round.

We consider the \emph{3-majority dynamics}:
At every round, every node pulls the opinion from three random neighbors and sets 
his new opinion to the majority one (ties are broken arbitrarily).
Let $k$ be the number of valid opinions. We show that, if $k \leqslant n^{\alpha}$, 
where $\alpha$ is a suitable positive constant, the 3-majority dynamics 
converges in time polynomial in $k$ and $\log n$ with high probability even in the presence of an adversary 
who can affect up to $o(\sqrt{n})$ nodes at each round. 

Previously, the convergence of the 3-majority protocol was known for $|\Sigma|
= 2$ only, with an argument that is robust to adversarial errors. On the other
hand, no anonymous,   uniform-gossip protocol that is robust to
adversarial errors was known for $|\Sigma| > 2$.

\end{abstract}


\thispagestyle{empty}
\newpage

\section{Introduction}
We study  the following probabilistic, {\em synchronous} process on a complete
  network of $n$ anonymous nodes: At the beginning, each node
holds an ``opinion'' which is an element of an arbitrary  finite set $\Sigma$. We 
call an opinion {\em valid} if it is held by at least one node at the
beginning. Then, in  each   round, the following happens: 1) every node pulls the    opinion  from three 
  random nodes and sets its new opinion to 
the majority one (ties are broken arbitrarily), and 
 2) an adaptive  \emph{dynamic adversary} can arbitrarily
change  the opinions of some  nodes.
We consider \emph{$F$-dynamic adversaries} that, at every round,  can change the opinions of up to $F$ nodes, possibly 
  introducing  non-valid opinions.

Let the system start from any configuration 
    having  $k$ valid opinions with $k \leq n^{\alpha}$ for some constant $\alpha<1$ and consider any 
    $F$-dynamic  adversary with $F = \bigO (\sqrt n / (k^{5/2}\log n))$.
 We prove that the process converges to a configuration in which all but $\mathit{O}(\sqrt{n})$ nodes
hold the same valid opinion within  $\mathit{O}( (k^{2} \sqrt{\log n} + k \log n ) (k+\log n) )$
 rounds, with high probability.   So, this 
 bounded adversary has no relevant chances to force the system to converge to non-valid opinions.

This shows that the \emph{3-majority dynamics}  provides an efficient solution
to the \emph{stabilizing-consensus} problem    in the \emph{uniform-gossip} model. 
Previously, this was known only for the binary case, i.e. $|\Sigma | = 2$, while for any  $| \Sigma| \geq 3$, 
it  has been an important open question for several years
   \cite{AAE07,DGMSS11}.  Furthermore, still for any  $| \Sigma| \geq 3$,
   $o(n)$-time convergence of the
3-majority dynamics  was open even in the absence of an adversary whenever
the initial bias toward some  plurality opinion is not large.

In the reminder of this section, we will   describe  in more detail the consensus problem and various
network scenarios  in which it is of interest, our result in this setting, and a comparison with 
previous related results.

\subsection{Consensus (or Byzantine agreement)}
The {\em consensus} problem in a distributed network is defined as follows: A
collection of agents, each holding a piece of information (an element
of a set $\Sigma$), interact with the goal of agreeing on one
of the elements of $\Sigma$ initially held by at least one agent,
possibly in the presence of an adversary that is trying to disrupt the
protocol. The consensus problem in the presence of an adversary (known
as Byzantine agreement) is a fundamental primitive in the design of
distributed algorithms \cite{PSL80,R83}.
The goal is to design a distributed, local protocol that
brings the system into  a configuration that meets the following
conditions: \emph{(1) Agreement}: All non-corrupted nodes
support  the same opinion $v$; \emph{(2) Validity}:  The 
opinion   $v$   must be a \emph{valid } one, i.e., an  opinion  
which was initially declared  by  at least one (non-corrupted) 
node; \emph{(3) Termination}: Every non-corrupted node can 
correctly decide to stop running the protocol at some round.

Recently, there has been considerable interest in the design of consensus
algorithms in models that severely restrict both communication and
computation \cite{AAE07,BCNPS15,DGMSS11}, both for efficiency consideration and because such
models capture aspects of the way consensus is reached in social
networks, biological systems, and other domains of interest in network
science \cite{AAD+06,AFJ06,Dolev,cardelli2012cell,Doty14,FHK14,HouseHunt}.

In particular, we assume an anonymous network in which nodes 
possess no unique IDs, nor do they have any static   binding   of their local link 
ports (i.e., nodes cannot keep track of \emph{who sent what}).  
From the point of view of computation, the most restrictive setting is
to assume that each node only has $\bigO (\log |\Sigma|)$ bits of memory 
available, i.e., it just suffices to store a constant 
number of opinions. We further assume that this bound extends to 
link bandwidth available in each round. Finally, communication 
capabilities are severely 
constrained and non-deterministic: Every node can communicate 
with at most a (small) constant number of random neighbors 
in each round.   
These constraints are well-captured by  the \emph{uniform-gossip}   
communication  model \cite{DGHILSSST87,KSSV00,KDG03}:  
At every round, every node can exchange a (short) message 
(say, $\Theta(\log(|\Sigma|))$ bits)  with each of at most  $h$  random 
neighbors, where $h$ is a (small) absolute constant\footnote
{In fact, $h = 1$ in the standard uniform-gossip model. 
It is easy to verify that all our results still hold 
in this more restricted model at the cost of a constant 
slow-down in convergence time and local memory size.}. A more recent,  
sequential  variant  of the uniform-gossip model is the \emph
{(random) population-protocols} model \cite{AAE07,AAE06,AAD+06} in which, in each 
round, a single interaction between a pair of randomly selected nodes 
occurs. 

\bigskip
\noindent The classic notion of consensus is too strong and unrealistic in the 
aforementioned distributed settings,  that instead rely on \emph{weaker} forms of 
consensus, deeply investigated in \cite{AAE07,AFJ06,A12,DGMSS11}. 
In this paper, we consider a variant of  the \emph
{stabilizing-consensus} problem \cite{AFJ06} considered   in  \cite{AAE07}:
  There, a solution is required to converge to a stable
\emph{regime}  in which the above  three properties   are guaranteed 
in  a relaxed,  still useful  form\footnote{ These 
relaxed convergence   properties  are described in detail in Section 7 of \cite{AAE07}.}. More precisely:
  
\begin{definition}
\label{def:stabcons}
A \emph{stabilizing almost-consensus} protocol must ensure the 
following properties: 

\noindent
- \emph{Almost  agreement.} Starting from any initial 
configuration, in a finite number of rounds, the system must reach 
a   \emph{regime} of       configurations where all but a \emph
{negligible} ``bad'' subset   (i.e. having size    $\bigO (n^{\gamma})$ 
for  constant $\gamma < 1$) of the nodes support the same 
opinion.    

\noindent 
- \emph{Almost validity.} The system is required to 
converge w.h.p. to an almost-agreement regime where all but a 
negligible  bad set  of nodes keep the same \emph{valid} opinion.

\noindent - \emph{Non termination.}  In    dynamic distributed 
systems,    nodes represent  simple and anonymous computing units 
which       are  not  necessarily    able to  detect  any global 
property.

\noindent - \emph{Stability.}  The convergence toward such a   
weaker form of agreement is only guaranteed to  hold \emph{ with 
high probability} (in short, \emph{w.h.p.}\footnote{According to 
the standard definition, we say that a sequence  of  events 
$\mathcal E_n$,  $n= 1,2, \ldots$ holds \emph{with high 
probability} if  $ \Prob{}{\mathcal E_n} = 1-\bigO (1/n^{\lambda})$ for  
some positive constant $\lambda >0$.}) and only over a \emph{long 
period} (i.e. for any arbitrarily-large polynomial number of 
rounds).
 
\end{definition}

The main result of this paper is on the convergence properties of the 
3-majority dynamics  in the uniform-gossip model in the presence of the    
adaptive $F$-dynamic adversary (defined above) and of the adaptive \emph{$F$-static adversar}y. In the latter,
the adversary looks at the initial configuration, then changes the opinion 
of up to $F$ nodes and,  after that, no further adversary's actions are allowed.  

\begin{theorem}
Let $k \leq n^{\alpha}$  for some constant  $\alpha <1 $  and $F = \adv$ for some constant   $\beta>0$. 
Starting from any initial configuration having $k$ valid opinions, the  3-majority
dynamics  reaches a (valid) stabilizing almost-consensus in presence of   any
$F$-dynamic adversary within 
$\bigO( (k^{2} \sqrt{\log n} + k \log n ) (k+\log n) )$
rounds, w.h.p. \\
Moreover, the same bound on the convergence time holds in the presence of  any 
  $F$-static adversary with a larger bound on $F$, i.e.,   $F = n/k-\sqrt{kn\log n}$. 
\end{theorem}

In \cite{BCNPST13}, an $\Omega(k \log n)$ bound on the  convergence-time of the 3-majority dynamics
is derived (that holds even  when the system  starts from biased configurations): So, our  bound is 
almost-tight whenever $k = \bigO(\polylog (n))$.

Not assuming a 
large initial bias of the plurality opinion considerably complicates the analysis. Indeed, 
the major open challenge  is the analysis  from 
(almost) uniform configurations, where     the system needs to break the initial symmetry 
in the absence of significant drifts towards any of the initial opinions.  So far, this issue has never been 
analyzed even in the non-adversarial case. Moreover, the 
phase before symmetry breaking is the one in which the adversary 
has more chances to cause undesired behaviours: Long   delays 
and/or convergence towards non-valid opinions. 
In Section \ref{sec:prely}, after providing some preliminaries, 
we shall discuss the above technical challenges.
 
\subsection{Previous results}
Consensus problems in distributed systems have been the focus of a 
large body of work in several research areas, such as distributed computing \cite{GK10},
communication networks \cite{RM08}, social networks and voting systems \cite{MNT14,YOASS13}, distributed databases \cite{DGHILSSST87,DGMMPR10}, biological systems and Chemical Reaction Networks \cite{cardelli2012cell}.
For brevity's sake, we here focus on results that are closest in 
spirit to our work. 

In \cite{AAE07}, the authors 
show that w.h.p. $n$ agents that meet at random can  reach valid stabilizing almost-consensus 
in $\bigO(n \log n)$ pairwise interactions against 
an $F= o(\sqrt{n})$-bounded   dynamic adversary.  The adopted protocol is the well-studied 
third-state protocol \cite{AAE07,PVV09}. However, their  analysis (and, thus, their  result) only holds for 
the binary case and
for the \emph{population-protocol}   model: At every round only one pair of nodes can interact.  
The authors left the existence of protocols for the multi-valued Byzantine case as a final 
open question \cite{AAE07}.
In general, sequential processes are much easier to analyze 
than parallel ones (like those yielded by the  uniform-gossip model): For instance, 
the resulting Markov chains are reversible \cite{lpw08} while those arising from parallel
processes are non-reversible. 
   
In  the uniform-gossip model, in  \cite{DGMSS11} the authors   provide an   analysis 
of the \emph{3-median} rule, in which every node   updates its  value to 
the median of its random sample.
They show that this dynamics   converges to  an almost-agreement configuration (which is 
even a good approximation of the global  median) within 
$\bigO(\log k \cdot \log\log  n + \log n)$ rounds, w.h.p.  
It turns out that, in  the binary case, the median rule is equivalent to    
the 3-majority dynamics, thus    their result implies that 3-majority 
 is   an $(F = \sqrt{n})$-stabilizing consensus   
with $\bigO (\log n)$ convergence time.  However, in the non-binary case,  
it requires $\Sigma$ to be a totally-ordered set and the possibility  to perform 
  basic algebraic operations: This is a rather strong restriction in applications arising from 
social networks, voting-systems, and bio-inspired systems. More importantly,
we emphasize  that, even   assuming  an  ordered opinion set $(\Sigma, \leq)$,
the 3-median rule   does not
guarantee  the crucial property of \emph{validity} against both $F$-static (and, clearly, dynamic)  adversaries even 
for very-small bounds on $F$ (say $F = \polylog(n)$).  

We strongly believe that the validity property of consensus  plays a crucial role
in several  realistic  scenarios, such as   monitoring sensor 
networks, bio-inspired dynamic systems, and voting systems \cite{cardelli2012cell,MNT14,YOASS13}.

More recently, the 3-majority rule in the multi-opinion case (i.e. for $|\Sigma| \geq 3$) has been studied for a  
stronger goal than consensus, namely, \emph{stabilizing plurality consensus} \cite{BCNPST13}. In this task, the goal 
is to reach an almost-stable consensus towards the   valid opinion initially
supported by \emph{the plurality} of the nodes. However, the initial configuration is assumed to have a
large bias towards the pluraltiy opinion.
Then, let $k$ be the number of valid opinions, and let $s$ be the initial 
difference between the largest and the second-largest opinion:  
By strongly exploiting    the assumption
$s \geq \sqrt{k n \log n}$, the authors in \cite{BCNPST13} proved that, w.h.p., the system converges   
to the plurality opinion within time $\Theta(k \log n)$.
    
Another  version of binary stabilizing almost-consensus  is the one 
studied by Yildiz et al in \cite{YOASS13}: Here, corrupted nodes 
are \emph{stubborn} agents of a social network  who influence 
others but never change their opinions. They prove  negative 
results under a generalized variant of the  classic voter dynamics in 
the (Poisson-clock) population-protocol model.

\section{The Process and its Analysis in a Nuthshell} \label{sec:prely}

\noindent
\textbf{Preliminaries.}
We assume a  distributed system consisting of $n$ nodes that 
communicate with  each other over a complete graph via the 
synchronous \emph{uniform-gossip} mechanism: In every 
round, each node can pull information from (at most) $h$ random 
neighbors, where $h$ is an absolute constant (in this work, $h = 3$).
At the onset, every node chooses an arbitrary item, called 
\emph{opinion}, from an arbitrary finite 
set $\Sigma$.  
A simple dynamics for consensus is the \emph{$3$-majority protocol} 
\cite{BCNPST13}:

\begin{quote}
\textit{In each round, every node  
samples three nodes uniformly at random (including itself and with repetitions) 
and revises its opinion  according to the majority  of the opinions it sees. 
If it sees three different opinions, it picks the first one.}
\end{quote}
\noindent
Clearly, in the case of three different opinions, choosing the second or 
the third one would not make any difference, nor would choosing one of 
the observed opinions uniformly at random.

Since the communication graph is complete and  nodes are anonymous, the 
overall system state  at any round can be described by a \emph{ 
configuration} $\mathbf{c}:=\left(c_{1},...,c_{|\Sigma|}\right)$, 
where  the \emph{support}  $c_{i}$ of opinion $i$  is  the number of nodes holding      opinion $i$ in 
that system's state. Given    configuration $\mathbf{c}$, we say that an opinion $i$ is \emph{active} in $\mathbf{c}$ if
$c_i >0$ and, 
for any set of active opinions $W\subseteq \Sigma$,     we define $\mincol[W]:=\arg\min_{i\in W}c_{i}$.
For any variable  $x$ of the process, we write 
$x^{\left(t\right)}$  if we  are considering  its value at round $t$ and  $X^{\left(t\right)}$ to
denote the corresponding random variable. 
Furthermore, following \cite{lpw08}, considered a configuration 
$\mathbf{c}$ and a random variable $X$ defined over the process, we 
write $\Prob{\mathbf{c}}{X^{(t)} = x}$ for $\Prob{}{X^{(t)} = x 
\, \vert \, \mathbf{C}^{(0)} = \mathbf{c}}$, i.e., to denote the probability 
distribution of the variable $X$ when the system evolves for $t$ 
consecutive rounds starting from configuration $\mathbf{c}$. 
Analogously, we write $\Expec{\mathbf{c}}{X^{(t)}}$ for the associated 
conditional expectation.

\bigskip 
The next lemma  provides 
 the expected number of nodes supporting a given opinion  at round $t+1$ (and a general 
  upper  bound to it), given the   configuration at
round $t$. The simple proof of the first equality is in 
\cite{BCNPST13}. It is also  included  in Appendix \ref{sec:apx:prely}  to make the paper 
 self-contained.    

\begin{lemma}[See \cite{BCNPST13}]
\label{lem:average} Let $\mathbf{c}$ be the   configuration at round  $t$ and let $W \subseteq \Sigma$ 
be the subset of active opinions   in $\mathbf{c}$.
 Then, for any opinion $i \in W$,  
\begin{equation}
\expec{ C_{i}^{(t+1)} }{\mathbf{C}^{(t)}= \mathbf{c} } = c_i \left(1+\frac{c_{i}}{n}-\frac{\sum_{j \in W}c_{j}^{2}}{n^{2}}\right)
\leq 
c_{i}\left(1+\frac{c_{i}}{n}-\frac{1}{|W|}\right)\label{eq:average}
\end{equation}
\end{lemma}

\noindent
The above upper bound easily   implies that 
  opinions whose supports fall  below the average $n/|W|$ decrease
   in expectation.   
   This expected drift  is a key-ingredient of our analysis and, as we will
   see in the next paragraph, it provides useful intuitions about the process.
On the other hand, when $\mathbf{c}$ is almost uniform, the above \emph{drift} 
turns out to be negligible and symmetry breaking is due to the inherent variance of the random process.

\smallskip
\noindent
\textbf{Failed attempts.}
When the 3-majority dynamics starts from      configurations that 
exhibit a large initial   support bias   between the largest and the second-largest 
opinions, the   approach  adopted in \cite{BCNPST13} successfully 
exploits the fact that the initial plurality is preserved throughout 
the evolution of the random process, with an expected  positive  
drift that is also preserved, w.h.p.  An intuition of this fact can be 
achieved from simple manipulations of (\ref{eq:average}). However, 
the aforementioned drift is only preserved if the largest opinion never changes 
(w.h.p.), \emph{no matter which the second-largest opinion is}: a 
condition that is not met by uniform configurations.
A promising attempt to cope with uniform configurations is to 
consider the r.v.  $S^{(t)}= C_{\imax(t)}^{(t)} - C_{\jmax(t)}^{(t)}$ 
where $\imax(t)$ and $\jmax(t)$ are the r.v.s   that take the 
index of (one of) the largest opinion and of (one of) the 
second-largest ones, respectively, in round $t$. For any  \emph{fixed} pair $i,j$, 
such that $c_i > c_j$, (\ref{eq:average}) implies that the   
difference $C^{(t+1)}_i-C^{(t+1)}_j$ in the next round is positive 
in expectation, so a suitable  submartingale argument \cite{lpw08} 
seemed to work in order to show that the system  (rather quickly) 
achieves  a ``sufficiently-large'' bias toward the plurality as to 
allow fast convergence. This approach  would work if  the 
\emph{random} indices  $\imax$ and $\jmax$ maintained their initial 
values across the entire duration of the process. Unfortunately, starting from uniform  
configurations, in the next round, the expected difference between 
the \emph{new} largest opinion and the \emph{new}  second largest 
one may have no positive drift at all. Roughly  speaking, in the next 
round,  the  r.v.  $C^{(t+1)}_{\jmax(t+1)}$ can be much larger than the r.v.
$C^{(t+1)}_{\imax(t)}$.
 
A promising dynamics for the stabilizing almost-consensus problem is the one introduced in \cite{DGMSS11}, 
in which nodes revise their 
opinions (assumed to be totally ordered) by taking the median between 
the currently held opinion and those held by two randomly 
sampled nodes. However, while we do not assume opinions to be  
integers (or totally ordered), their analysis strongly relies on the fact that the 
median opinion (or any good approximation of it) exhibits a strong increasing drift, 
even when starting from almost-uniform configuration, whereas no 
opinion is ``special'' to a majority rule when the starting configuration is uniform. 
The adoption of an inherently biased function as the median can have 
important consequences. To get an intuition, the reader 
may consider the following simple   instance: $\Sigma = \{1,2,3\}$, 
with the system starting in configuration $c_1= n/2, c_2=0, 
c_3= n/2$.  At the end of the first round, a static adversary 
changes the values  of $F = \log n$ nodes, equally distributed in 
$c_1$ and $c_3$, to value 2. The (non-valid) value 2 is the \emph
{global median} and some counting arguments show that, while 
values $1$ and $2$ have no positive expected drift, the median has 
an exponential expected drift that holds w.h.p. whenever $c_1,c_2 = 
\Theta(n)$. This might fool the system into the 
configuration in which $c_2 = n$, thus converging to a non-valid 
value. 
  
  \smallskip   
  \noindent
  \textbf{Our New Approach: An Overview.}
Our analysis  significantly departs from the above approaches.
It is important to remark that, for  $|\Sigma| \geq 3$,  no analysis of the 
3-majority dynamics with almost-uniform initial configurations is  
known, even in the simpler non-adversarial case.  
On the other hand, while simpler, the analysis of the non-adversarial 
case still has \emph{per-se} interest and it requires to address 
some of the main technical challenges that also arise in the adversarial 
case. Section \ref{se:noadv} will be thus  devoted to the analysis of the 
non-adversarial case, 
while an outline is given in the paragraphs that  follow. 

 When the   configuration is (approximately) uniform, Lemma 
 \ref{lem:average} says  us that 
the process exhibits no significant drift toward any 
\emph{fixed} opinion. Interestingly, things change if we consider the 
random variable $C_{\mathbbm{m}}^{(t)}$, indicating the smallest 
opinion support at round $t$. 
Let $j\leq k$ be the number of active opinions in 
a given round $t$,  we first   prove that the  expected value of $C_{\mathbbm{m}}^{(t)}$ 
always exhibits a non-negligible negative drift:

\begin{equation}\label{eq:cmindrift-road}
        \Expec{}{C_\mathbbm{m}^{(t+1)} \,|\, \mathbf{C}^{(t)} =
        \hat{\mathbf{c}}} \leqslant c_{\mathbbm{m}} - \varepsilon
        \frac{\sqrt{n}}{j^{3/2}} \, , \mbox{ for some constant } \varepsilon >0  \end{equation}
This drift  is  
essentially a consequence of Lemma \ref{lem:average} \emph{and} of 
the standard deviation of r.v.s  $C_{i}^{(t)}$s
 (see  the proof of Lemma \ref{lemma:breaking}). 
The analysis then proceeds along consecutive phases, each 
consisting of a suitable number of consecutive rounds. If the 
number of active opinions at the beginning of the generic phase is $j$, we 
prove that, with positive constant probability, $C_{\mathbbm{m}}^{(t)}$  
vanishes within the end of the phase, so that the next phase begins 
with (at most) $j-1$ active opinions. 

\noindent
We   clearly need   a  good   bound on the   length of a phase 
beginning with at most $j$ opinions. To 
this aim, we derive  
  a new upper  bound - stated in Lemma \ref{lemma:expectimewithdrift} - on the  \emph{hitting 
time} of stochastic processes with    expected drift  that are defined by    finite-state  Markov 
chains \cite{lpw08}. Thanks to this result, we can use  the negative drift in 
(\ref{eq:cmindrift-road}) to  
  prove that, from any configuration with $j \leqslant k$ active opinions,   
$C_{\mathbbm{m}}^{(t)}$  drops below the threshold $n/j - \sqrt{j n 
\log n}$ within $\mathcal{O}(\poly(j,\log n))$ rounds, with  
constant positive probability:   This ``hitting'' event represents     the exit condition 
 from    the \emph
{\em symmetry-breaking stage} of the phase. Indeed, once it  occurs,  we can consider
   \emph{any fixed} active opinion $i$  having   support size   $c_i$  below the above threshold
   (thanks to the previous stage, we know that 
   there is a good chance this opinion    exists): We then show that    $C_i$
has a         negative drift  of order $\Omega(c_i/j)$.  This allows us 
  to  prove that $C_i$ drops from
$n/j - \sqrt{j n 
\log n}$ to zero    within     $\mathcal{O}(\poly(j,\log n))$ further rounds,
 with positive constant probability.
This interval of rounds is the \emph {\em dropping stage} of the phase.  

\noindent
Ideally, the process proceeds along $k$ consecutive phases, indexed as  
$j= k, k-1, \ldots, 2$, such that we are left with at most $j-1$ active 
opinions at the end of Phase $j$. In practice, we only have a constant probability 
that at least one opinion disappears during Phase $j$. However, using 
standard probabilistic arguments, we can prove that, w.h.p., for every  $j$, the 
transition from $j$ to $j-1$ active opinions takes a constant (amortized) 
number of phases, each requiring $\mathcal{O}(\poly(j,\log n))$ rounds.

The presence of a dynamic, adaptive  adversary makes the above analysis 
technically more complex. A major issue is that 
a different definition of \emph{Phase} must be considered, since the adversary might permanently
feed any  opinion so that the latter    never dies. So the number of active opinions might not decrease
from one phase to the next one. Essentially, we need to manage the persistence of ``small'' (valid or not)
opinions: The end of a phase is now characterized by one ``big'' valid color that becomes ``small'' and, moreover, 
we need to show that, in general,  ``small'' colors never 
becomes ``big'', no matter what the dynamic  $F$-bounded adversary does.
An informal description of the dynamic-adversary case is given in Subsection \ref{ssec::dynadvs}.

\section{The 3-Majority Dynamics without Adversary}\label{se:noadv}
Let $\goodcols \subseteq \Sigma$ be the subset of valid opinions, i.e. those supported by at 
least one node in the initial configuration, and denote by $k = |\goodcols|$
its size. This section is devoted to the proof of the following result.

\begin{theorem}[The Adversary-Free Case.]
\label{theorem:final}
Starting from any initial configuration with $k \leqslant n^{1/3 - \varepsilon}$ active opinions, 
where $\varepsilon > 0$ is an arbitrarily-small constant, the 3-majority dynamics reaches
consensus within 
$\mathcal{O}\left( (k^{2} \log^{1/2} n + k \log n) (k + \log n) \right)$ rounds, w.h.p.
\end{theorem}

\smallskip\noindent 
We first provide the lemmas required for the process analysis and then we give the formal proof of  the above theorem.

The next lemma shows an upper bound on the time it takes a stochastic process 
with values in $N = \{0,1,\dots, n\}$ to reach or exceed a target value $m$, under mild hypotheses on the process.
We here give only an idea of the proof, the full proof is in Appendix~\ref{sec:apx:prely}.

\begin{lemma}\label{lemma:expectimewithdrift}
Let $\{X_t\}_t$ be a Markov chain with finite state space $\Omega$, let $f : \Omega \rightarrow N$ be a function mapping states of the chain in non-negative integer numbers, and let $\{Y_t\}_t$ be the stochastic process over $N$ defined by $Y_t = f(X_t)$. Let $m \in N$ be a ``target value'' and let 
$$
\tau = \inf \{ t \in \mathbb{N} \,:\, Y_t \geqslant m \}
$$
be the random variable indicating the first time $Y_t$ reaches or exceeds value $m$.
Assume that, for every state $x \in \Omega$ with $f(x) \leqslant m-1$, it holds that
\begin{enumerate}
\item (Positive drift). $\Expec{}{Y_{t+1} \,|\, X_t = x} \geqslant f(x) + \lambda$ for some $\lambda > 0$ 
\item (Bounded jumps).  $\Prob{x}{Y_{\tau} \geqslant \alpha m} \leqslant \alpha m /n$, for some $\alpha > 1$.
\end{enumerate}
Then, for every starting state $x \in \Omega$, it holds that
$$
\Expec{x}{\tau} \leqslant 2 \alpha \frac{m}{\lambda}
$$
\end{lemma}
\ideaproof
From Hypothesis \textit{1} it follows that $Z_t = Y_t - \lambda t$ is a \emph{submartingale} that satisfies the hypotheses of the Doob's \emph{Optional Stopping Theorem} \cite{Doob53} (see e.g. Corollary~17.8 in~\cite{lpw08} or Theorem~10.10 in~\cite{Williams91}), thus
\[
0 \leqslant f(x) = \Expec{x}{Z_0} \leqslant \Expec{x}{Z_{\tau}} = \Expec{x}{Y_\tau} - \lambda \Expec{x}{\tau}
\]
And from Hypothesis \textit{2} it follows that $\Expec{x}{Y_\tau} \leqslant 2 \alpha m$.
\qed

\medskip\noindent
We now exploit the above lemma
in order to bound the time required by the \emph{symmetry-breaking} stage.  

\begin{lemma}[Symmetry-breaking stage]\label{lemma:breaking}
Let $\mathbf{c}$ be any configuration with $j$ active opinions. Within $t = \mathcal{O}\left(j^2 \log^{1/2} n\right)$ rounds it holds that 
$$
\Prob{\mathbf{c}}{\exists i \mbox{ such that } C_i^{(t)} \leqslant n/j - \sqrt{j n \log n}} \geqslant \frac{1}{2}
$$
\end{lemma}

\skproof
Let $J$ be the set of $j$ active opinions in   $\mathbf{c}$ and let $\col^{(t)} = \left( C_i^{(t)}\,:\, i \in J \right)$ be the random variable indicating the opinion configuration at round $t$, where we assume $\col^{(0)} = \mathbf{c}$. Let $C_{\mathbbm{m}}^{(t)} = \min\left\{ C_i^{(t)} \,:\, i \in J \right\}$ be the minimum among  all $C_i^{(t)}$s and consider the stochastic process $\{Y_t\}_t$ defined as 
$
Y_t = \left\lfloor n/j\right\rfloor - C_{\mathbbm{m}}^{(t)}$.
Observe that $Y_t$ takes values in $\left\{0,1, \dots, \lfloor n/j \rfloor \right\}$ and it is a function of      
 $\col^{(t)}$. We are interested in the first time $Y_t$ becomes at least as large as $\sqrt{j n \log n}$, i.e.
$$
\tau = \inf \left\{ t \in \mathbb{N} \,:\, Y_t \geqslant \sqrt{j n \log n} \right\}
$$
We now show that $\{Y_t\}_t$ satisfies Hypotheses \textit{1} and \textit{2} of Lemma~\ref{lemma:expectimewithdrift},
with $\lambda = \varepsilon \sqrt{n}/j^{3/2}$, for a suitable constant $\varepsilon > 0$.

\smallskip\noindent
\textit{1.} Let $\hat{\mathbf{c}} = \left( \hat{c}_i \,:\, i \in J \right)$ be any configuration with $j$ active opinions such that $\hat{c}_{\mathbbm{m}} > n/j - \sqrt{j n \log n}$. We want to prove that
\begin{equation}\label{eq:cmindrift}
\Expec{}{C_\mathbbm{m}^{(t+1)} \,|\, \mathbf{C}^{(t)} = \hat{\mathbf{c}}} \leqslant c_{\mathbbm{m}} - \varepsilon \frac{\sqrt{n}}{j^{3/2}}
\end{equation}

\noindent
Two cases may arise.

\smallskip\noindent
\underline{Case $\hat{c}_{\mathbbm{m}} > n/j - 2 \varepsilon \sqrt{n/j}$:} 
  Observe that, in this case, 
  r.v.s $\left\{C_i^{(t+1)} \,:\, i \in J \right\}$ conditional on $\{C^{(t)} = \hat{\mathbf{c}} \}$  have standard deviation  $\Omega\left(\sqrt{n/j}\right)$. Moreover, they
   are binomial and negatively associated.  Hence,  by choosing $\varepsilon$ small enough, 
   from the Central Limit Theorem we have that
$$
\Prob{}{i \in J \mbox{ exists such that } C_i^{(t+1)} \leqslant \frac{n}{j} - 6 \varepsilon \cdot \sqrt{\frac{n}{j}}} \geqslant 1/2
$$
We thus get 
\begin{equation}\label{eq:cmincase2}
\Expec{}{C_\mathbbm{m}^{(t+1)} \,|\, \mathbf{C}^{(t)} = \hat{\mathbf{c}}} 
\leqslant \frac{1}{2} \left( \frac{n}{j} - 6 \varepsilon \cdot \sqrt{\frac{n}{j}} \right) + \frac{1}{2} \cdot \frac{n}{j} 
= \frac{n}{j} - 3 \varepsilon \sqrt{\frac{n}{j}} 
\leqslant c_{\mathbbm{m}} - \varepsilon \sqrt{\frac{n}{j}} 
\leqslant c_{\mathbbm{m}} - \varepsilon \frac{\sqrt{n}}{j^{3/2}}
\end{equation}

\smallskip\noindent
\underline{Case $\hat{c}_{\mathbbm{m}} \leqslant n/j - 2 \varepsilon \sqrt{n/j}$:} 
Equation \eqref{eq:cmindrift} easily follows from Lemma \ref{lem:average}. Indeed, let $i \in J$ be an opinion such that $\hat{c}_i = \hat{c}_{\mathbbm{m}}$, then
\begin{eqnarray}
\Expec{}{C_\mathbbm{m}^{(t+1)} \,|\, \mathbf{C}^{(t)} = \hat{\mathbf{c}}} 
& \leqslant & \Expec{}{C_i^{(t+1)} \,|\, \mathbf{C}^{(t)} = \hat{\mathbf{c}}} 
\leqslant \hat{c}_i \left( 1 + \frac{\hat{c}_i}{n} - \frac{1}{j} \right) \nonumber \\
& \leqslant & \hat{c}_i \left( 1 - \frac{2\varepsilon}{\sqrt{nj}} \right) 
\leqslant \hat{c}_i - \frac{\varepsilon\sqrt{n}}{j^{3/2}} 
= \hat{c}_{\mathbbm{m}} - \varepsilon \frac{\sqrt{n}}{j^{3/2}} \label{eq:negdrif2}
\end{eqnarray}
where  we used the case's condition   and the fact that    $\hat{c}_i = \hat{c}_{\mathbbm{m}} \geqslant n/(2j)$.

\smallskip\noindent
\textit{2.} Since random variables $\left\{C_i^{(t+1)} \,:\, i \in J \right\}$
conditional on
the configuration at round $t$ are binomial, it is possible to apply Chernoff
bound (though with some care) to prove that
\begin{equation}\label{eq:hypjumpnew}
\Prob{\mathbf{c}}{Y_{\tau} \geqslant \alpha \sqrt{j n \log n}} \leqslant
\frac{1}{n}, \
\mbox{ for some constant } \ \alpha > 1
\end{equation}
Though this result seems intuitive, its formal proof is less obvious,
since $\tau$ is a stopping time and thus itself a random variable.
Lemma \ref{le:stophp} in Appendix \ref{app::nonadv} offers a formal proof of
the above statement.

\smallskip\noindent
From \eqref{eq:cmindrift} and \eqref{eq:hypjumpnew}, we have that $\{Y_t\}_t$ satisfies the hypotheses of~Lemma~\ref{lemma:expectimewithdrift} with $m = \sqrt{j n \log n}$ and $\lambda = \varepsilon \sqrt{n}/j^{3/2}$. Hence $\Expec{\mathbf{c}}{\tau} < j^2 \sqrt{\log n}$ and, from Markov inequality, for $t = 2 j^2 \sqrt{\log n}$, we finally get 
$$
\Prob{\mathbf{c}}{ \forall \   i \in J   :  C_i^{(t)} \geqslant n/j - \sqrt{j n \log n}} \leqslant \Prob{\mathbf{c}}{\tau > 2 j^2 \sqrt{\log n}} \leqslant \frac{1}{2}
$$
\qed

\medskip\noindent
We  now  provide the analysis of the \emph{dropping} stage: More precisely, we show  that, if the system starts with 
 up to  $j$ active  opinions and one of them (say $i$) is   below the threshold 
$n/j - \sqrt{j n \log n}$, then $i$  drops to   the smaller threshold $j^2 \log n$  within $\mathcal{O}(j \log n)$ additional rounds.  
This bound can be proved w.h.p. since, in this regime, 
$C_i$ is still  sufficiently large to apply the Chernoff bound. This concentration result
is not  necessary to the purpose of 
proving Theorem \ref{theorem:final}, while it is a key ingredient in 
the analysis of the adversarial case (Theorem \ref{theorem:final-adv}).  
The next   lemma can be proved by standard concentration arguments - applied in an iterative way -
on the r.v. $C_i^{(t)}$ (see  Appendix \ref{app::nonadv}).

\begin{lemma}[Dropping stage 1]\label{le:dying_noadv}
Let $\mathbf{c}$ be any configuration with $j \leqslant n^{1/3 - \varepsilon}$
active opinions, where $\varepsilon > 0$ is an arbitrarily-small positive constant,
and such that an opinion $i$ exists with $c_i \leqslant n/j - \sqrt{j n \log
n}$. Within $t = \mathcal{O}(j \log n)$ rounds opinion $i$ becomes
$\mathcal{O}\left(j^2 \log n\right)$ w.h.p.
\end{lemma}

\smallskip\noindent 
In the next lemma we prove that once $c_i$ becomes smaller than $n/(2j)$, 
then opinion $i$ disappears within further $\mathcal{O}(j \log n)$ rounds with constant probability.
We here give only an idea of the proof, the full proof is in Appendix~\ref{app::nonadv}.

\begin{lemma}[Dropping stage 2]\label{le:dying_noadv2}
Let $\mathbf{c}$ be any configuration with $j \leqslant n^{1/3 - \varepsilon}$ active opinions, where $\varepsilon > 0$ is an arbitrarily-small positive constant, and such that an opinion $i$ exists with $c_i \leqslant n/(2j)$. Within $t = \mathcal{O}(j \log n)$ rounds opinion $i$ disappears with probability at least $1/2$.
\end{lemma}
\ideaproof
If $c_i \leqslant n/(2j)$ in configuration $\mathbf{c}$, then from Lemma \ref{lem:average} it follows that
\[
  \Expec{}{C_i^{(t+1)} \,|\, \mathbf{C}^{(t)} = \mathbf{c}} \leqslant c_i \left(1 - \frac{1}{2j} \right)
\]
Moreover, since $C_i^{(t+1)}$ conditional on $\left\{\mathbf{C}^{(t)} = \mathbf{c}\right\}$ is binomial, if $j \leqslant n^{1/3-\varepsilon}$, from the Chernoff bound it follows that
$\Prob{}{C_i^{(t+1)} > n/(2j) \,|\, \mathbf{C}^{(t)} = \mathbf{c}} \leqslant e^{-\Theta(n^{\varepsilon})}$. Hence, it is easy to check that for any initial configuration $\mathbf{c}$ with $c_i \leqslant n/(2j)$ the following recursive relation holds
\[
\Expec{\mathbf{c}}{C_i^{(t)}} \leqslant \left(1 - \frac{1}{2j} \right) \Expec{\mathbf{c}}{C_i^{(t-1)}} + e^{-n^{\varepsilon/2}} 
\]
that for some $t = \mathcal{O}(j \log n)$ gives $\Expec{\mathbf{c}}{C_i^{(t)}} \leqslant 1/2$. Since $C_i^{(t)}$ is a non-negative integer-valued r.v., the thesis then follows from the Markov inequality.
\qed

\medskip\noindent
\textbf{Proof of Theorem \ref{theorem:final}}. From Lemmas \ref{lemma:breaking}, \ref{le:dying_noadv}, and \ref{le:dying_noadv2} it follows that from any configuration with $j \leqslant k$ active opinions, within $\mathcal{O}(k^2 \sqrt{\log n} + k \log n)$ rounds at least one of the opinions disappears with probability at least $1/4$. Thus, within $\mathcal{O}((k^2 \sqrt{\log n} + k \log n) (k + \log n))$ rounds, all opinions but one disappear w.h.p.
\qed

\section{Convergence Time of 3-Majority with Adversary}

In this section we consider the presence of a Byzantine adversary that can
adaptively change the opinion of a bounded number of nodes in order to delay
convergence time toward a valid consensus, or even worse, to let the system
converge toward a non valid one. We consider two different adversarial
strategies: A static one and a stronger, dynamic one.

\subsection{The \texorpdfstring{$F$}{F}-static adversary}
At the end of the first round, once every node has fixed his own initial
opinion, the adversary looks at the configuration and arbitrarily replaces
the opinion of at most $F = n/k-\sqrt{kn\log n}$ nodes with an arbitrary opinion
in $\Sigma$. Then the protocol starts the process and no further adversary's
actions  are allowed. Since any opinion the adversary may introduce has size 
less than $n/k-\sqrt{k n \log n}$, as a simple consequence of the 
dropping stage 
(see  Lemmas \ref{le:dying_noadv} and \ref{le:dying_noadv2}), the static adversarial case easily reduces to the
non-adversarial one. We thus get the following 

\begin{corollary}
	Let $ k \leq n^{\alpha}$ for some constant $\alpha<1$ and $F = n/k-\sqrt{kn\log n}$. Starting from any initial configuration having
	$k$ opinions, the  3-majority protocol  reaches a stabilizing
	almost-consensus in presence of any $F$-static adversary within
	$\mathcal{O}\left( k^2 \sqrt{\log n} + k\log n \right)(k + \log n)$ rounds, w.h.p. 
 \label{cor:static-adv}
\end{corollary}

\subsection{The \texorpdfstring{$F$}{F}-dynamic adversary} 
 \label{ssec::dynadvs}
 The actions of this    adversary  over the studied process can be described  as follows.
At the end of \emph{every round} $t$, after nodes have updated their opinions 
(i.e. once the configuration $ \mathbf{C}^{(t)} = \mathbf{c}^{(t)}$ is realized),
the $F$-dynamic adversary looks at  the current opinion configuration and replaces
the opinion of up to $F$ nodes with any opinion in $\Sigma$. 
 
In what follows we consider an $F$-dynamic adversary with $F \leqslant 
\adv$ for a suitable positive constant $\beta$.  
As we will show in the proof of Lemma \ref{le:breaking-adv}, this  bound on $F$ turns
out to be almost tight for guaranteeing that the process converges to an 
almost-consensus regime in polynomial time, w.h.p.

The presence of the adversary requires us to distinguish between valid and 
non valid opinions.
So, we recall that the  set of valid opinions $\goodcols \subseteq \Sigma$ is the subset of active 
opinions in the initial configuration
and  we observe that, in the reminder of this section, 
$k$ denotes  the number of valid opinions, i.e.,  $k : = |\goodcols|$.

We are now ready to state our main result  in the presence of the dynamic
adversary (its full proof is given in Appendix   \ref{sec:adv}). 

\begin{theorem}[The Dynamic-Adversary Case.]
\label{theorem:final-adv}
Let $k \leq n^{\alpha}$  for some constant  $\alpha <1 $  and $F = \adv$ for some constant   $\beta>0$.
Starting from any initial configuration having $k$ opinions, the  3-majority
    reaches a (valid) stabilizing almost-consensus in presence of   any
$F$-dynamic adversary within 
$O( (k^{2} \sqrt{\log n} + k \log n ) (k+\log n) )$
rounds, w.h.p.
\end{theorem}

\begin{proof}[Idea of the  Proof]
We here provide a description of the main technical differences w.r.t.
the analysis  for the non-adversarial case.

\noindent
As discussed in the overview of the process analysis, the adversary can introduce  ``small''    non-valid opinions  
  and  it can keep  small  valid opinions active  that would otherwise disappear (as shown in Section \ref{se:noadv}).
These facts lead us to  the problem of  managing ``small'' opinions: The rigorous definition of
small opinion   
is determined by    the minimal negative drift for $C_{\mathbbm{m}}^{(t)}$ we derived 
in the proof of Lemma \ref{lemma:breaking} (see (\ref{eq:negdrif2})).

\noindent
Let 
$ \bad:=\{ i\ \big|\ c_{i}\leq\smallsize\} $ 
be the set of \emph{small opinions} 
for some constant $\smallsizeconst > \advconst$, 
and let its complement
$B:=\bar{\bad} = \{ i\ \big|\ c_{i}>\smallsize\} $
be the set of \emph{big opinions}.
 
 \noindent
 It turns out that we cannot
use the definition of the (end of) phase adopted in the non-adversarial case: At least one (valid) opinion
 dies. Wlog, let us assume that, at the beginning,  all the $k$ valid opinions 
are big. Then the new phase $j$ is  an interval of consecutive rounds, in each of which
exactly $j$ big valid opinions are present. The new  goal is to show that at the end of phase $j$, one of the $j$ big colors will get 
small and, moreover, this color (and no other small color) will never get  big.  In the symmetry-breaking stage of each phase, 
we thus need to show that the negative drift of $C_{\mathbbm{m}}^{(t)}$
 (notice that the latter 
 now denotes the minimum    among  the
  $j$  big colors) cannot be opposed by the actions of the $F$-dynamic adversary, provided
that $F \leq \adv$. This fact  (stated in 
Lemma \ref{le:breaking-adv}) is obtained via  two different technical steps:  
i) A new bound on the expected   negative drift for $C_{\mathbbm{m}}^{(t)}$ that  considers both   
  the presence of small good opinions and the adversary's opposing action (this result is formalized in 
  Lemma \ref{lem:average-adv});
  ii) A novel use of Lemma \ref{lemma:expectimewithdrift} on the hitting time of random processes  
   in order to  bound the expect time of   the symmetry-breaking stage. We in fact need to define a new 
  stopping condition that also includes some ``bad'' event: Some small (valid or not)  color become big.
    We then 
        show
       that    bad stopping events  never happen along the entire process, w.h.p. (this is essentially guaranteed 
       by     Lemma \ref{le:hyp-H}).
       
      \noindent
      The dropping stage of phase $j$   is now defined as the interval of rounds in which 
      $C_{\mathbbm{m}}^{(t)}$ drops from the symmetry-breaking   threshold $n/j - \sqrt{j n \log n}$ 
           to the size of small colors i.e. $\smallsize$. Similarly to the non-adversary case, we can here fix 
           the big opinion $i$ that is dropped below the symmetry-breaking   threshold and look at its 
           negative drift derived from     Lemma \ref{lem:average-adv}. The drift is strong enough to tolerate the 
             the actions of the $F$-bounded adversary and implies an $O(j \log n)$   bound  on the time required
             by this second stage of phase $j$. This stage's analysis is given in Lemma \ref{le:dying-adv}.

\noindent
Finally,   after $k$ phases, we are left with one
(valid) opinion that accounts for $n - \mathit{O}(\sqrt{n})$ nodes, while the
remaining nodes can have any (possibly non valid) opinion and reflect the
presence of the adversary. In fact, this is what happens with high probability.

\end{proof}

\section{Future Works}
We strongly believe that our  upper bound on the convergence time of the 3-majority dynamics is not tight w.r.t.
  $k$. The factor    $\Omega(k^3)$  seems     to be not  necessary:
 We believe that  at least a factor $k$ can be saved. To this aim, 
 we would need to show that  ``more''  opinions get small during a  phase.  This number should  also depend on the
 current number   of big   colors.  Another idea would be that of (also) considering the growth of the maximal
  opinion.
  Unfortunately, differently from the minimal  opinion (see (\ref{eq:cmindrift-road}) in Section \ref{sec:prely}), we don't have 
 any   good bound on the  expected drift   for the maximal opinion that holds from  \emph{any} configuration. So, we don't see
 how to efficiently adapt our approach without this crucial ingredient.

\newpage
\bibliographystyle{plain}
\bibliography{sb}

\begin{thebibliography}{10}

\bibitem{AAE06}
Dana Angluin, James Aspnes, and Eisenstat David.
\newblock {Stably computable predicates are semilinear}.
\newblock In {\em In Proc. of the 25th Ann. ACM SIGACT-SIGOPS Symp. on
  Principles of Distributed Computing (PODC'06)}, pages 292--299. ACM, 2006.

\bibitem{AAD+06}
Dana Angluin, James Aspnes, Zo{\"{e}} Diamadi, Michael~J. Fischer, and Peralta
  Ren{\'e}.
\newblock {Computation in networks of passively mobile finite-state sensors}.
\newblock {\em Distributed Computing}, 18(4):235--253, 2006.

\bibitem{AAE07}
Dana Angluin, James Aspnes, and David Eisenstat.
\newblock {A Simple Population Protocol for Fast Robust Approximate Majority}.
\newblock {\em Distributed Computing}, 21(2):87--102, 2008.
\newblock (Preliminary version in DISC'07).

\bibitem{AFJ06}
Dana Angluin, Michael~J. Fischer, and Hong Jiang.
\newblock {Stabilizing consensus in mobile networks}.
\newblock In {\em Proc. of Distributed Computing in Sensor Systems (DCOSS'06)},
  volume 4026 of {\em LNCS}, pages 37--50, 2006.

\bibitem{A12}
James Aspnes.
\newblock {Faster Randomized Consensus with an Oblivious Adversary}.
\newblock In {\em Proc. of the 31st Ann. ACM SIGACT-SIGOPS Symp. on Principles
  of Distributed Computing (PODC'12)}, pages 1--8. ACM, 2012.

\bibitem{BCNPS15}
Luca Becchetti, Andrea Clementi, Emanuele Natale, Francesco Pasquale, and
  Riccardo Silvestri.
\newblock {Plurality Consensus in the Gossip Model}.
\newblock In {\em Proc. of the 26th Ann. ACM-SIAM Symp. on Discrete Algorithms
  (SODA'15)}, pages 371--390. SIAM, 2015.

\bibitem{BCNPST13}
Luca Becchetti, Andrea Clementi, Emanuele Natale, Francesco Pasquale, Riccardo
  Silvestri, and Luca Trevisan.
\newblock {Simple dynamics for plurality consensus}.
\newblock In {\em Proc. of the 26th ACM Symp. on Parallelism in Algorithms and
  Architectures (SPAA'14)}, pages 247--256. ACM, 2014.

\bibitem{Dolev}
Ohad Ben-Shahar, Shlomi Dolev, Andrey Dolgin, and Michael Segal.
\newblock {Direction election in flocking swarms}.
\newblock In {\em Proc. of the 6th Int. Workshop on Foundations of Mobile
  Computing (DIALM-POMC'10)}, pages 73--80. ACM, 2010.

\bibitem{cardelli2012cell}
Luca Cardelli and Attila Csik{\'a}sz-Nagy.
\newblock {The Cell Cycle Switch Computes Approximate Majority}.
\newblock {\em Scientific Reports}, Vol. 2, 2012.

\bibitem{DGHILSSST87}
Alan Demers, Dan Greene, Carl Hauser, Wes Irish, John Larson, Scott Shenker,
  Howard Sturgis, Dan Swinehart, and Doug Terry.
\newblock {Epidemic algorithms for replicated database maintenance}.
\newblock In {\em Proc. of the 6th Ann. ACM Symposium on Principles of
  Distributed Computing (PODC'12)}, pages 1--12. ACM, 1987.

\bibitem{DGMMPR10}
Martin Dietzfelbinger, Andreas Goerdt, Michael Mitzenmacher, Andrea Montanari,
  Rasmus Pagh, and Michael Rink.
\newblock {Tight thresholds for cuckoo hashing via XORSAT}.
\newblock In {\em Proc. of the 37th Int. Coll. on Automata, Languages, and
  Programming (ICALP'10)}, volume 6198 of {\em LNCS}, pages 213--225. Springer,
  2010.

\bibitem{DGMSS11}
Benjamin Doerr, Leslie~A. Goldberg, Lorenz Minder, Thomas Sauerwald, and
  Christian Scheideler.
\newblock {Stabilizing consensus with the power of two choices}.
\newblock In {\em Proc. of the 23rd Ann. ACM Symp. on Parallelism in Algorithms
  and Architectures (SPAA'11)}, pages 149--158. ACM, 2011.

\bibitem{Doob53}
Joseph~L. Doob.
\newblock {\em Stochastic Processes}.
\newblock John Wiley \& Sons Inc., 1953.

\bibitem{Doty14}
David Doty.
\newblock {Timing in chemical reaction networks}.
\newblock In {\em Proc. of 25th Ann. ACM-SIAM Symp. on Discrete Algorithms
  (SODA'14)}, pages 772--784. SIAM, 2014.

\bibitem{FHK14}
Ofer Feinerman, Bernhard Haeupler, and Amos Korman.
\newblock {Breathe Before Speaking: Efficient Information Dissemination Despite
  Noisy, Limited and Anonymous Communication}.
\newblock In {\em Proc. of the ACM Symposium on Principles of Distributed
  Computing (PODC '14)}. ACM, 2014.

\bibitem{HouseHunt}
Nigel~R. Franks, Stephen~C. Pratt, Eamonn~B. Mallon, Nicholas~F. Britton, and
  David~J.T. Sumpter.
\newblock {Information flow, opinion polling and collective intelligence in
  house--hunting social insects}.
\newblock {\em Philosophical Transactions of the Royal Society of London B:
  Biological Sciences}, 357(1427):1567--1583, 2002.

\bibitem{GK10}
Seth Gilbert and Dariusz Kowalski.
\newblock {Distributed agreement with optimal communication complexity}.
\newblock In {\em Proc. of 21st Ann. ACM-SIAM Symp. on Discrete Algorithms
  (SODA'10)}, pages 965--977. SIAM, 2010.

\bibitem{KSSV00}
Richard Karp, Christian Schindelhauer, Scott Shenker, and Berthold Vocking.
\newblock {Randomized rumor spreading}.
\newblock In {\em Proc. of the 41th Ann. IEEE Symp. on Foundations of Computer
  Science (FOCS'00)}, pages 565--574. IEEE, 2000.

\bibitem{KDG03}
David Kempe, Alin Dobra, and Johannes Gehrke.
\newblock {Gossip-Based Computation of Aggregate Information}.
\newblock In {\em Proc. of 43rd Ann. IEEE Symp. on Foundations of Computer
  Science (FOCS'03)}, pages 482--491. IEEE, 2003.

\bibitem{lpw08}
David Levin, Yuval Peres, and Elizabeth~L. Wilmer.
\newblock {\em {Markov Chains and Mixing Times}}.
\newblock American Mathematical Society, 2008.

\bibitem{MNT14}
Elchanan Mossel, Joe Neeman, and Omer Tamuz.
\newblock {Majority dynamics and aggregation of information in social
  networks}.
\newblock {\em Autonomous Agents and Multi-Agent Systems}, 28(3):408--429,
  2014.

\bibitem{PSL80}
Marshall Pease, Robert Shostak, and Leslie Lamport.
\newblock {Reaching agreement in the presence of faults}.
\newblock {\em Journal of the ACM}, 27(2):228--234, 1980.

\bibitem{PVV09}
Etienne Perron, Dinkar Vasudevan, and Milan Vojnovic.
\newblock {Using Three States for Binary Consensus on Complete Graphs}.
\newblock In {\em Proc. of the 28th IEEE Conf. on Computer Communications
  (INFOCOM'09)}, pages 2527--1535. IEEE, 2009.

\bibitem{R83}
Michael~O. Rabin.
\newblock {Randomized byzantine generals}.
\newblock In {\em Proc. of the 24th Ann. Symp. on Foundations of Computer
  Science (SFCS'83)}, pages 403--409. IEEE, 1983.

\bibitem{RM08}
Yongxiang Ruan and Yasamin Mostofi.
\newblock {Binary Consensus with Soft Information Processing in Cooperative
  Networks}.
\newblock In {\em Proc. of the 47th IEEE Conf. on Decision and Control
  (CDC'08)}, pages 3613--3619. IEEE, 2008.

\bibitem{Williams91}
David Williams.
\newblock {\em Probability with Martingales}.
\newblock Cambridge University Press, 1991.

\bibitem{YOASS13}
Mehmet~E. Yildiz, Asuman~E. Ozdaglar, Daron Acemoglu, Amin Saberi, and Anna
  Scaglione.
\newblock {Binary Opinion Dynamics with Stubborn Agents}.
\newblock {\em ACM Trans. Econ. Comput.}, 4(1), 2013.

\end{thebibliography}

\newpage
\appendix

\section{Preliminary Results}\label{sec:apx:prely}

\paragraph{Proof of Lemma \ref{lem:average}}
According to the $3$-majority protocol, a node $u$ gets opinion $i$ if it chooses 3 times opinion $i$, or if it chooses two times $i$ and one time a different opinion, or if it chooses the first time opinion $i$ and then, the second and third time, two different distinct opinions.
Hence, if we denote by $X_{i,u}^{(t)}$ the indicator random 
variable of the event ``Node $u$ gets opinion $i$ at time $t$'', we 
have that

\vspace{-2mm}
\begin{small}
\begin{eqnarray*}
  \Prob{}{X_{i,u}^{(t+1)}=1 \,|\, \mathbf{C^{(t)}} = \mathbf{c}}
& = & \left(\frac{c_{i}}{n}\right)^{3}+3\left(\frac{c_{i}}{n}\right)^{2}\left(\frac{n-c_{i}}{n}\right)+\left(\frac{c_{i}}{n}\right)\left[1-\left(\frac{\sum_{\ell \in S}^{k}c_{\ell}^{2}}{n^{2}}+2\left(\frac{c_{i}}{n}\right)\left(\frac{n-c_{i}}{n}\right)\right)\right]\\
& = & \left(\frac{c_{i}}{n^{3}}\right)\left(n^{2}+c_{i}n-\sum_{\ell \in S}^{k}c_{\ell}^{2}\right)
\end{eqnarray*}
\end{small}

\vspace{-2mm}\noindent
Then the   inequality in \eqref{eq:average}  is obtained by observing that the sum 
$\sum_{\ell \in S} c_{\ell}^{2}$ is minimized for $c_{\ell}  = n/|S|$.

\qed

\paragraph{Proof of Lemma \ref{lemma:expectimewithdrift}}

Consider the stochastic process $Z_t = Y_t - \lambda t$ and observe that for any state $x \in \Omega$ with $f(x) \leqslant \border -1$ it holds that
\begin{eqnarray*}
\Expec{}{Z_{t+1} \,|\, X_t = x} & = & \Expec{}{Y_{t+1} \,|\, X_t = x} - \lambda (t+1) \\
& \geqslant & f(x) + \lambda - \lambda (t+1) \\
& \geqslant & f(x)  - \lambda t
\end{eqnarray*}
where in the inequality we used Hypotheses \textit{1}. Thus $Z_t$ is a \emph{submartingale} up to the stopping time $\tau$, i.e. $\Expec{}{Z_{t+1} \,|\, X_t} \geqslant Z_t$ for any $t < \tau$. Moreover, since $|Y_t| \leqslant n$ the \emph{jumps} of $Z_t$ can be bounded by a value independent of $t$
$$
|Z_{t+1} - Z_t| = \left| Y_{t+1} - \lambda(t+1) - Y_t + \lambda t \right| \\
\leqslant n + \lambda
$$
and it is easy to see that Hypotheses~\textit{1} implies $\Expec{x}{\tau} < \infty$, thus we can apply \emph{Doob's Optional Stopping Theorem}~\cite{Doob53} (see also, e.g., Corollary~17.8 in~\cite{lpw08} and Theorem~10.10 in~\cite{Williams91}). It then follows that $\Expec{x}{Z_{\tau}} \geqslant \Expec{x}{Z_0} = f(x)$ and, since $\Expec{x}{Z_{\tau}} = \Expec{x}{Y_{\tau}} - \lambda \Expec{x}{\tau}$, we have that
$$
\Expec{x}{\tau} \leqslant \frac{\Expec{x}{Y_{\tau}} - f(x)}{\lambda} 
\leqslant \frac{\Expec{x}{Y_{\tau}}}{\lambda}
$$
Finally, we get
\begin{eqnarray*}
\Expec{0}{Y_{\tau}} & = & \sum_{j = 1}^n j \Prob{0}{Y_{\tau} = j} \\
& = & \sum_{j = 1}^{\lfloor \alpha \border  \rfloor} j \Prob{0}{Y_{\tau} = j} + \sum_{j = \lfloor \alpha \border  \rfloor + 1}^n j \Prob{0}{Y_{\tau} = j} \\
& \leqslant & \left( \alpha \border  \right) + n \Prob{0}{Y_{\tau} > \alpha \border } \leqslant 2 \left( \alpha \border  \right)
\end{eqnarray*}
where in the last inequality we used Hypothesis~\textit{2}.
\qed

 \section{Proofs for the Non-Adversarial Case} \label{app::nonadv}
 
 \paragraph{Proof of Lemma \ref{le:dying_noadv}}
We first prove that the decreasing rate of $C_i$ depends on its value at the
end of the previous round.
More formally, if we are in a configuration satisfying the hypotheses of
the lemma:

\[
\Prob{}{C_{i}^{(t)}>c^{(t-1)}_{i}\left(1-\frac{1}{2}\left(\frac{1}{j}-\frac{c^{(t-1)}_{i}}{n}\right)\right)}
= \Prob{}{C_{i}^{(t)}>c^{(t-1)}_{i}\left(1-\left(\frac{1}{j}-\frac
{c^{(t-1)}_{i}}{n}\right)\right)(1 + \delta)}, 
\]
where
$ \delta = {\frac{1}{2}(\frac{1}{j}-\frac{c^{(t-1)}_{i}}{n})} 
/ {1-(\frac{1}{j}-\frac
{c^{(t-1)}_{i}}{n})}$

Using Lemma \ref{lem:average} and applying Chernoff bound we have:
\begin{eqnarray}
    \Prob{}{C_{i}^{(t)}>c^{(t-1)}_{i}\left(1-\frac{1}{2}
    \left(\frac{1}{j}-\frac{c^{(t-1)}_{i}}{n}\right)\right)} 
    &\leq& \exp\left\{
        -\frac{\delta^2}{3}\left(1-\left(\frac{1}{j} 
        - \frac{c^{(t-1)}_{i}}{n}\right)\right)c^{(t-1)}_{i}\right\}
        \nonumber\\
    &=& \exp\left\{ 
    -\frac{\delta}{3}\left(\frac{1}{2}\left(\frac{1}{j} 
    - \frac{c^{(t-1)}_{i}}{n}\right)\right)c^{(t-1)}_{i}\right\}\nonumber\\
    &<& \exp\left\{ 
    -\frac{1}{3}\left(\frac{1}{2}\left(\frac{1}{j} 
    - \frac{c^{(t-1)}_{i}}{n}\right)\right)^2c^{(t-1)}_{i}\right\}\nonumber\\
    &=& n^{-\Theta\left(1\right)}.
\label{eq:hp-noadv}
\end{eqnarray}
The second equality in \eqref{eq:hp-noadv} follows from the definition of $\delta$, 
while the third inequality follows by (upper) bounding the denominator 
of $\delta$ by $1$, which is always possible since $c_i/n - 1/j < 0$ 
from the hypotheses. Finally, to prove the last equality, we used the fact that   
 $c_i \geqslant j^2 \log n$ and that the function $x\left(1-x\right)^{2}$
is decreasing iff $x\in\left(1/3,1\right)$, with $x=jc_{i}/n$.

Finally, we  can iteratively apply \eqref{eq:hp-noadv} as long as we have at 
most $j$  active opinions and  $C^{(t)}_i$ keeps not smaller than  $j^2 \log n$. 
By standard concentration arguments we get that the time to reach this threshold is 
$\mathcal{O}\left(j \log n\right)$,  w.h.p.
\qed

\paragraph{Proof of Lemma \ref{le:dying_noadv2}}

Let $J$ be the set of active opinions. By conditioning on all the configurations $\hat{\mathbf{c}} = \left( \hat{c}_\ell \,:\, \ell \in J \right)$ that the system can take at round $t-1$, we can bound the expectation of $C_i^{(t)}$ as follows
\begin{eqnarray*}
\Expec{\mathbf{c}}{C_i^{(t)}} 
& = & \sum_{\hat{\mathbf{c}}} \Expec{}{C_i^{(t)} \,|\, \mathbf{C}^{(t-1)} = \hat{\mathbf{c}}} \Prob{\mathbf{c}}{\mathbf{C}^{(t-1)} = \hat{\mathbf{c}}} \\
& \leqslant & \left(1 - \frac{1}{2j} \right) \sum_{\hat{\mathbf{c}} \,:\, \hat{c}_i \leqslant n/(2j)} \hat{c}_i \cdot \Prob{\mathbf{c}}{\mathbf{C}^{(t-1)} = \hat{\mathbf{c}}} + n \cdot \sum_{\hat{\mathbf{c}} \,:\, \hat{c}_i > n/(2j)} \Prob{\mathbf{c}}{\mathbf{C}^{(t-1)} = \hat{\mathbf{c}}} \\
& \leqslant & \left(1 - \frac{1}{2j} \right) \Expec{\mathbf{c}}{C_i^{(t - 1)}} + n \cdot \Prob{\mathbf{c}}{C_i^{t-1} > \frac{n}{2j}}
\end{eqnarray*}
where we used that, for any configuration $\hat{\mathbf{c}}$ with $\hat{c_i} \leqslant n/(2j)$, Lemma \ref{lem:average} gives the bound $\Expec{}{C_i^{(t)} \,|\, \mathbf{C}^{(t-1)} = \hat{\mathbf{c}}} \leqslant \hat{c}_i \left(1 - \frac{1}{2j} \right)$.
Moreover, if $j \leqslant n^{1/3 - \varepsilon}$, from Chernoff bound it follows that
$$
\Prob{}{C_i^{(t)} > \frac{n}{2j} \,|\, \mathbf{C}^{(t-1)} = \hat{\mathbf{c}}} \leqslant e^{-\Theta\left(n^\varepsilon\right)}
$$
for any such configuration $\hat{\mathbf{c}}$. Hence, for any $t$ we have that $\Prob{\mathbf{c}}{C_i^{(t)} > \frac{n}{2j}} \leqslant t e^{- \Theta(n^{\varepsilon})}$. Indeed,
\begin{eqnarray*}
\Prob{\mathbf{c}}{C_i^{(t)} > \frac{n}{2j}} & \leqslant & \Prob{\mathbf{c}}{\exists \bar{t} = 1, \dots, t \,:\, C_i^{(\bar{t})} > \frac{n}{2j} \wedge C_i^{(\bar{t}-1)} \leqslant \frac{n}{2j}} \\
& \leqslant & \sum_{\bar{t} = 1}^{t} \Prob{\mathbf{c}}{C_i^{(\bar{t})} > \frac{n}{2j} \wedge C_i^{(\bar{t}-1)} \leqslant \frac{n}{2j}} \\
& = & \sum_{\bar{t} = 1}^{t} \, \sum_{\hat{\mathbf{c}} \,:\, \hat{c}_i \leqslant n/(2j)} \Prob{}{C_i^{(\bar{t})} > \frac{n}{2j} \,|\, \mathbf{C}^{(\bar{t}-1)} = \hat{\mathbf{c}}} \Prob{\mathbf{c}}{\mathbf{C}^{(\bar{t}-1)} = \hat{\mathbf{c}}} \\
& \leqslant & t e^{-\Theta\left( n^\varepsilon \right)}
\end{eqnarray*}
Thus for any $t = \poly(n)$ the following recursive relation holds
$$
\Expec{\mathbf{c}}{C_i^{(t)}} \leqslant \left(1 - \frac{1}{2j} \right) \Expec{\mathbf{c}}{C_i^{(t-1)}} + e^{-n^{\varepsilon/2}} 
$$
And it gives
$$
\Expec{\mathbf{c}}{C_i^{(t)}} \leqslant \left( 1 - \frac{1}{2j} \right)^t \frac{n}{2j} + e^{-n^{\varepsilon/3}}
$$
Hence, for $t = 2 j (\log n + 1)$ we have that $\Expec{\mathbf{c}}{C_i^{(t)}} \leqslant 1/2$ and
 since $C_i^{(t)}$ takes non-negative integer values, the thesis follows from Markov inequality.
\qed

\begin{lemma}\label{le:stophp}
Let $\mathbf{c}$ be any configuration with $j$ active opinions.  
Consider the stochastic process $\{Y_t\}_t$ defined as 
$Y_t = \left\lfloor\frac{n}{j}\right\rfloor - C_{\mathbbm{m}}^{(t)}$ 
and define the stopping time $\tau = \inf \left\{ t \in \mathbb{N} 
\,:\, Y_t \geqslant \sqrt{j n \log n} \right\}$. Then:
\[
\Prob{\mathbf{c}}{Y_\tau > \alpha\sqrt{jn\log n}}\le\frac{1}{n}.
\]
\end{lemma}
\begin{proof}
First of all, $\Expec{\mathbf{c}}{\tau} < \infty$, since 
$C_{\mathbbm{m}}^{(t)}$ has a negative drift (see the proof of Lemma 
\ref{lemma:breaking}). Next, from the definition of $Y_t$:
\begin{eqnarray*}
  \Prob{\mathbf{c}}{Y_\tau > \alpha\sqrt{jn\log n}} & = & 
\Prob{\mathbf{c}}{C_{\mathbbm{m}}^{(\tau)} < 
\left\lfloor\frac{n}{j}\right\rfloor - \alpha\sqrt{j n \log n}} \\
& = & \Prob{\mathbf{c}}{\exists\ell: C_{\ell}^{(\tau)} < 
\left\lfloor\frac{n}{j}\right\rfloor - \alpha\sqrt{j n \log n}} \\
& \le & \sum_{\ell = 1}^j\Prob{\mathbf{c}}{C_{\ell}^{(\tau)} <
\left\lfloor\frac{n}{j}\right\rfloor - \alpha\sqrt{j n \log n}}.
\end{eqnarray*}
Next, from the definition of the stopping time $\tau$:
\begin{small}
\begin{eqnarray}
&&\Prob{\mathbf{c}}{C_{\ell}^{(\tau)} <
\left\lfloor\frac{n}{j}\right\rfloor - \alpha\sqrt{j n \log n}} = \sum_{t = 
1}^\infty\Prob{\mathbf{c}}{C_{\ell}^{(t} <
\left\lfloor\frac{n}{j}\right\rfloor - \alpha\sqrt{j n \log n}\bigwedge\tau 
= t} = \nonumber\\
&& = \sum_{t = 1}^\infty\Prob{\mathbf{c}}{C_{\ell}^{(t} <
\left\lfloor\frac{n}{j}\right\rfloor - \alpha\sqrt{j n \log 
n}\bigwedge C_{\mathbbm{m}}^{(t)} \le
\left\lfloor\frac{n}{j}\right\rfloor - \sqrt{j n \log n} \Bigg\vert
\bigwedge_{s = 1}^{t-1}C_{\mathbbm{m}}^{(t)} >
\left\lfloor\frac{n}{j}\right\rfloor - \sqrt{j n \log n}}\nonumber\\
&& \quad\cdot\Prob{\mathbf{c}}{\bigwedge_{s = 1}^{t-1}C_{\mathbbm{m}}^{(t)} 
> \left\lfloor\frac{n}{j}\right\rfloor - \sqrt{j n \log n}} = \nonumber\\
&& = \sum_{t = 1}^\infty\Prob{\mathbf{c}}{C_{\ell}^{(t} <
\left\lfloor\frac{n}{j}\right\rfloor - \alpha\sqrt{j n \log 
n}\Bigg\vert
\bigwedge_{s = 1}^{t-1}C_{\mathbbm{m}}^{(t)} >
\left\lfloor\frac{n}{j}\right\rfloor - \sqrt{j n \log n}} \nonumber\\
&& \quad\cdot\Prob{\mathbf{c}}{\bigwedge_{s = 1}^{t-1}C_{\mathbbm{m}}^{(t)} >
\left\lfloor\frac{n}{j}\right\rfloor - \sqrt{j n \log n}}\label{eq:9}
\end{eqnarray}
\end{small}
where the last equality follows since $\left (C_{\ell}^{(t} <
\left\lfloor\frac{n}{j}\right\rfloor - \alpha\sqrt{j n \log 
n}\right )$ implies $\left(C_{\mathbbm{m}}^{(t)} <
\left\lfloor\frac{n}{j}\right\rfloor - \sqrt{j n \log n}\right)$.

We next consider $\Prob{\mathbf{c}}{\bigwedge_{s = 
1}^{t-1}C_{\mathbbm{m}}^{(t)} >
\left\lfloor\frac{n}{j}\right\rfloor - \sqrt{j n \log n}}$. We can 
write:

\vspace{-2mm}
\begin{small}
\begin{eqnarray*}
\Prob{\mathbf{c}}{\bigwedge_{s = 1}^{t-1}C_{\mathbbm{m}}^{(t)} >
\left\lfloor\frac{n}{j}\right\rfloor - \sqrt{j n \log n}} 
& = & \prod_{s = 1}^{t-1}\Prob{\mathbf{c}}{C_{\mathbbm{m}}^{(s)} >
\left\lfloor\frac{n}{j}\right\rfloor - \sqrt{j n \log 
n}\Bigg\vert\bigwedge_{r = 1}^{s-1}C_{\mathbbm{m}}^{(r)} >
\left\lfloor\frac{n}{j}\right\rfloor - \sqrt{j n \log n}} \\
& = & \prod_{s = 1}^{t-1}\Prob{\mathbf{c}}{C_{\mathbbm{m}}^{(s)} >
\left\lfloor\frac{n}{j}\right\rfloor - \sqrt{j n \log 
n}\Bigg\vert C_{\mathbbm{m}}^{(s-1)} >
\left\lfloor\frac{n}{j}\right\rfloor - \sqrt{j n \log n}}\label{eq:10}
\end{eqnarray*}
\end{small}

\vspace{-2mm}\noindent
where the last equality follows since the $3$-majority process is 
Markovian. We next give an upper bound on 
$\Prob{\mathbf{c}}{C_{\mathbbm{m}}^{(s)} >
\left\lfloor\frac{n}{j}\right\rfloor - \sqrt{j n \log 
n}\Bigg\vert C_{\mathbbm{m}}^{(s-1)} >
\left\lfloor\frac{n}{j}\right\rfloor - \sqrt{j n \log n}}$:
\begin{small}
\begin{eqnarray*}
&&\Prob{\mathbf{c}}{C_{\mathbbm{m}}^{(s)} >
\left\lfloor\frac{n}{j}\right\rfloor - \sqrt{j n \log 
n}\Bigg\vert C_{\mathbbm{m}}^{(s-1)} >
\left\lfloor\frac{n}{j}\right\rfloor - \sqrt{j n \log n}} = \\
&& = \sum_{\mathbf{\hat{c}}: \hat{c}_{\mathbbm{m}} >
\left\lfloor\frac{n}{j}\right\rfloor - \sqrt{j n \log 
n}}\Prob{\mathbf{\hat{c}}}{C_{\mathbbm{m}}^{(1)} >
\left\lfloor\frac{n}{j}\right\rfloor - \sqrt{j n \log 
n}}\cdot \Prob{\mathbf{c}}{\mathbf{C}^{(s-1)} = 
\mathbf{\hat{c}}\Bigg\vert C_{\mathbbm{m}}^{(s-1)} >
\left\lfloor\frac{n}{j}\right\rfloor - \sqrt{j n \log n}}\le\\
&&\le \sum_{\mathbf{\hat{c}}: \hat{c}_{\mathbbm{m}} >
\left\lfloor\frac{n}{j}\right\rfloor - \sqrt{j n \log 
n}}\Prob{\mathbf{\hat{c}}}{C_{l}^{(1)} >
\left\lfloor\frac{n}{j}\right\rfloor - \sqrt{j n \log 
n}}\cdot \Prob{\mathbf{c}}{\mathbf{C}^{(s-1)} = 
\mathbf{\hat{c}}\Bigg\vert C_{\mathbbm{m}}^{(s-1)} >
\left\lfloor\frac{n}{j}\right\rfloor - \sqrt{j n \log n}},
\end{eqnarray*}
\end{small}
where $l = \arg\hat{c}_{\mathbbm{m}}$ (ties broken arbitrarily).
We can give an upper bound on  
$\Prob{\mathbf{\hat{c}}}{C_{l}^{(1)} >
\left\lfloor\frac{n}{j}\right\rfloor - \sqrt{j n \log 
n}}$ using a ``reverse'' Chernoff bound\footnote{A folklore example 
with complete proofs can be found at 
\url{http://cstheory.stackexchange.com/questions/14471/reverse-chernoff-bound}.}. 
In particular, it is possible to show that 
\begin{small}
\begin{eqnarray*}
&&\Prob{\mathbf{\hat{c}}}{C_{l}^{(1)} >
(1 - \delta)\Expec{\mathbf{\hat{c}}}{C_{l}^{(1)}}}\le 1 - e^{-\beta\delta^2\Expec{\mathbf{\hat{c}}}{C_{l}^{(1)}}}
\end{eqnarray*}
\end{small}
for a suitable constant $\beta$. We use $\delta = \sqrt{j n \log 
n}/\Expec{\mathbf{\hat{c}}}{C_{l}^{(1)}}$ and note that 
$n/2j\le\Expec{\mathbf{\hat{c}}}{C_{l}^{(1)}}\le n/j$, so that 
\begin{eqnarray}
&&\Prob{\mathbf{\hat{c}}}{C_{l}^{(1)} >
(1 - \delta)\Expec{\mathbf{\hat{c}}}{C_{l}^{(1)}}}\le 1 - 
e^{-4\beta j^2\log n}.\label{eq:11}
\end{eqnarray}
Saturating with respect to $\mathbf{\hat{c}}$ yields 
\[
\Prob{\mathbf{c}}{C_{\mathbbm{m}}^{(s)} >
\left\lfloor\frac{n}{j}\right\rfloor - \sqrt{j n \log 
n}\Bigg\vert C_{\mathbbm{m}}^{(s-1)} >
\left\lfloor\frac{n}{j}\right\rfloor - \sqrt{j n \log n}}\le 1 - 
e^{-4\beta j^2\log n}.
\]
On the other hand, using standard concentration techniques and 
recalling that $\Expec{\mathbf{\hat{c}}}{C_{l}^{(1)}}\le n/2j$, we can 
prove that:
\begin{eqnarray}
\Prob{\mathbf{c}}{C_{\ell}^{(t} <
\left\lfloor\frac{n}{j}\right\rfloor - \alpha\sqrt{j n \log 
n}\Bigg\vert
\bigwedge_{s = 1}^{t-1}C_{\mathbbm{m}}^{(t)} >
\left\lfloor\frac{n}{j}\right\rfloor - \sqrt{j n \log n}}\le 
e^{-\frac{\alpha^2}{6} j^2\log n}\label{eq:12}
\end{eqnarray}
Next, substituting \eqref{eq:11} and \eqref{eq:12} into \eqref{eq:9}, 
the result follows after simple calculations and by suitably choosing 
$\alpha$.
\end{proof}

\section{Convergence Time of 3-Majority with Adversary}
\label{sec:adv}

In this section we give the formal definitions of the notions introduced  in 
Section \ref{ssec::dynadvs} and we provide the  proof of  the lemmas stated 
 in the proof outline of Theorem \ref{theorem:final-adv}.

The actions of the dynamic adversary over the studied process can be
formalized as follows.

\begin{definition}\label{def:dynadv}
At the end of \emph{every round} $t$, after nodes have updated their opinions
(i.e. once the configuration $ \mathbf{C}^{(t)} = \mathbf{c}^{(t)}$ is
realized), the $F$-dynamic adversary looks at  the current opinion
configuration and replaces the opinion of up to $F$ nodes with any opinion in
$\Sigma$. We define    $\tilde { \mathbf{C} }^{(t)} $ as the configuration that
results from the adversary's action on $\mathbf{c}^{(t)}$  and
$\advi[\left(t\right)] = \advi[\left(t\right)]( \mathbf{c}^{(0)}, \tilde {
\mathbf{c} }^{(0)} , \dots , \mathbf{c}^{(t-1)}, \tilde { \mathbf{c} }^{(t
- 1)}, \mathbf{c}^{(t)} )$ as the r.v. corresponding to the number of
nodes that the adversary adds or removes from $c_{i}$ (note that $ \sum_{
i \in \Sigma } \left|\advi\right|\leq 2 F$) at the end of the $t$-th round,
based on all the past history of the process, i.e.
\[
    \tilde { \mathbf{C} }^{(t)} = \left( {C}_{1}^{(t)} +
    \advletter_{1}^{(t)}, \dots , { C }_{|\Sigma|}^{(t)} +
    \advletter_{|\Sigma|}^{(t)} \right)
\]
\end{definition}

We consider an $F$-dynamic adversary with $F \leqslant 
\adv$ for a suitable positive constant $\beta$.  
As we will show in the proof of Lemma \ref{le:breaking-adv}, this  bound on $F$ turns
out to be almost tight for guaranteeing that the process converges to an 
almost-consensus regime in polynomial time, w.h.p.

The presence of the adversary requires us to distinguish between valid and non
valid opinions. So, we recall that the  set of valid opinions $\goodcols
\subseteq \Sigma$ is the subset of active opinions in the initial configuration
and  we observe that, in the reminder of this section, $k$ denotes  the number
of valid opinions, i.e.,  $k : = |\goodcols|$. As discussed in the overview of
the process analysis (Section \ref{sec:prely}), the adversary can introduce  ``small''    non-valid
opinions  and  it can keep active small valid opinions that would otherwise
disappear (as shown in Section \ref{se:noadv}). These facts lead us to  the
problem of  managing ``small'' opinions, whose formal definition can be given
as follows.
 
\begin{definition}
Let 
$ \bad:=\{ i\ \big|\ c_{i}\leq\smallsize\} $ 
be the set of the \emph{small opinions}, where 
 $\smallsizeconst$ is some constant such that $\smallsizeconst > \advconst$, 
and let its complement $B:=\bar{\bad} = \{ i\ \big|\ c_{i}>\smallsize\} $ be the set of the  \emph{big opinions}.
\end{definition}

We now re-state Theorem \ref{theorem:final-adv}, 
whose idea of proof is outlined in Section \ref{ssec::dynadvs}. 

\begin{reptheorem}{theorem:final-adv}
  Let $k \leq n^{\alpha}$  for some constant  $\alpha <1 $  and $F = \adv$ for some constant   $\beta>0$.    Starting from any initial configuration having $k$ opinions, the  3-majority
     reaches a (valid) stabilizing almost-consensus in presence of   any
    $F$-dynamic adversary within 
    $O( (k^{2} \sqrt{\log n} + k \log n ) (k+\log n) )$
    rounds, w.h.p.
\end{reptheorem}

\begin{proof}[Sketch of proof]
    From Lemmas \ref{le:breaking-adv} and \ref{le:dying-adv} it follows that
    from any configuration with $j \leqslant k$ active opinions, within
    $\mathcal{O}(k^2 \sqrt{\log n} + k \log n)$ rounds at least one of the
    opinions becomes small with probability at least $1/2$. Thus within
    $\mathcal{O}((k^2 \sqrt{\log n} + k \log n) (k + \log n))$ rounds all
    opinions but one become small w.h.p.
\end{proof}

In order to prove Lemma \ref{le:breaking-adv} and Lemma \ref{le:dying-adv},
we need the following three lemmas.

\begin{lemma}
	\label{lem:average-adv} 
	Let $ \tilde{ \mathbf{c} }$ be any configuration such that
	$|\bigc|\leq j$ and $\sum_{i \in \badcols}	\tilde{c}_i^{(t)} \leq \smallsize$.
    For some constant $\smallbiasconst>0$, 
    for any opinion $i$ such that $ \tilde{ c }_{i} \geq \smallsize$,  it holds
	\begin{eqnarray}
		\expec{ 	{C}_{i}^{(t+1)} }
        { \tilde{ \mathbf{C} }^{(t)}=  \tilde{ \mathbf{c} } } 
		&\leq&
        \tilde{c}_{i}\left( 1 -  
         \frac 1j + \frac{ \tilde{c}_{i} + \smallbias}{n} 
        \right)
        \label{eq:average-before-adv} \\
		\expec{ 	\tilde{C}_{i}^{(t+1)} }
        { \tilde{ \mathbf{C} }^{(t)}=  \tilde{ \mathbf{c} } } 
		&\leq&
        \tilde{c}_{i}\left( 1 - \min \left\{ 
         \frac 1j - \frac{ \tilde{c}_{i} + \smallbias}{n} \ ,\ 
        \frac 12 \left( \frac 1j - \frac{ \tilde{c}_{i}}{n} \right) \right\}
        \right)
        \label{eq:average-adv}
    \end{eqnarray}
\end{lemma}
\begin{proof}
	Similarly to the proof of Lemma \ref{lem:average} we have 
	\begin{eqnarray*}
		&\expec{ C_{i}^{(t+1)} }
        { \tilde{ \mathbf{C} }^{(t)} =  \tilde{ \mathbf{c} } } 
		&\leq
         \tilde{ c }_{i}\left(1+\frac{ \tilde{ c }_{i}}{n}-\frac{\sum_{j}
         \tilde{ c }_j^2}{n^2}\right)
		\leq
         \tilde{ c }_{i}\left(1+\frac{\tilde{c}_{i}}{n}-\frac{ \sum_{j\in
         \bigc}\tilde{c}_j^2}{n^2}\right)\\
		&&\leq 
		 \tilde{ c }_{i}\left(1+\frac{\tilde{c}_{i}}{n}-\frac{ \sum_{j\in
		\bigc}\left(\frac{n -(k-j+1)\smallsize}{j}\right)^2}{n^2}\right)\\
		&&\leq
		\tilde{c}_{i}\left(1+\frac{\tilde{c}_{i}}{n}-\frac{ \sum_{j\in
        \bigc} ( n - \smallbiasconst / 4 \sqrt{ n / k } )^2}{ j^2 n^2 }\right)\\
		&&\leq
        \tilde{c}_{i}\left( 1+\frac{\tilde{c}_{i}}{n}-\frac 1j + \frac{
            \smallbiasconst / 2 \sqrt{ n / k } }{jn}\right)
		\leq          
        \tilde{c}_{i}\left( 1 - \frac{n / j - \tilde{c}_{i} - 
        \smallbiasconst / 2 \sqrt{ n / k }}{n}\right)
	\end{eqnarray*}
    Taking into account {any} possible action of the adversary, 
        we thus get that
        \begin{eqnarray}
            \expec{  \tilde{ C }_{i}^{(t+1)} }
            { \tilde{ \mathbf{C} }^{(t)} =  \tilde{ \mathbf{c} } } 
             &=&
            \expec{ C_{i}^{(t+1)} }
            { \tilde{ \mathbf{C} }^{(t)} =  \tilde{ \mathbf{c} } } 
            +
            \expec{ \advletter_{i}^{(t+1)} }
            { \tilde{ \mathbf{C} }^{(t)} =  \tilde{ \mathbf{c} } } \nonumber\\
            &\leq& 
            \tilde{c}_{i}\left( 1 - \frac{n / j - \tilde{c}_{i} - 
            \smallbiasconst / 2 \sqrt{ n / k }}{n}\right) + F \nonumber\\
            &\leq& 
            \tilde{c}_{i}\left( 1 - \frac{n / j - \tilde{c}_{i}  }{ n }
            + \frac { 2\max \left\{ \smallbiasconst / 2 \sqrt{ n / k }, 
                F n / \tilde{c}_{i} \right\}}{n}\right). 
            \label{eq:last-average-adv}
        \end{eqnarray}
    By distinguishing the cases $\tilde{c}_{i} \geq n / (3j)$ or 
    $\tilde{c}_{i} < n / (3j)$, from \eqref{eq:last-average-adv} we get
    \eqref{eq:average-adv}.
\end{proof}

In Lemma \ref{le:hyp-H} we prove the following   
key-properties  of the process in the presence of the dynamic adversary: 
w.h.p., it is \emph{never} the case that 
\begin{enumerate}
    \item if in a given round a valid opinion is small then it gets big 
        at a later time, i.e. $\bad^{(t-1)} \subseteq \bad^{(t)}$;
    \item  the size of the overall set of non valid opinions grows beyond 
        $\smallsize$, i.e. $\sum_{i \in \badcols}c_i \leq \smallsize$.
\end{enumerate}

\begin{lemma}
	\label{le:hyp-H}
	If $ \tilde{ \mathbf{c} }^{(t)}$ is such that 
    $\sum_{i \in \badcols} \tilde{ c }_i^{(t)} \leq \smallsize$, then 
    $\sum_{i \in \badcols} \tilde{ C }_i^{(t+1)} \leq \smallsize$ and 
    $\bad^{(t)} \subseteq \bad^{(t+1)}$, w.h.p.
\end{lemma}
\begin{proof}
	From Lemma \ref{lem:average-adv}, for each $i\in\bad^{(t)}$ we have that 
	\begin{eqnarray*} 
		\expec{  { C }_{i}^{(t+1)} }{ \tilde{ \mathbf{C} }^{(t)}
        =  \tilde{ \mathbf{c} } } 
        \leq  \tilde{ c }_{i}\left( 1+\frac{ \tilde{ c }_{i}}{n} 
        - \frac 1k \right) .
	\end{eqnarray*}
	From a direct application of the Chernoff bound to $C_{i}^{(t+1)} $, 
	and taking into account {any} possible action of the adversary, 
    we thus get that w.h.p. 
    \[
        \tilde{ C }_{i}^{(t+1)} =  C_{i}^{(t+1)} 
        + \advletter_i^{(t+1)}\leq\fsmallsize
         \left(1 - \frac 1{ 4k }\right) + F
		\leq \fsmallsize,
    \]
	that is, $i\in \bad^{(t)}$ w.h.p.
	Analogously, we have 
	\begin{eqnarray*} 
		\expec{ \sum_{i \in \badcols^{(t)}} { C }_i^{(t+1)} }
        { \tilde{ \mathbf{C} }^{(t)}=  \tilde{ \mathbf{c} } } 
		\leq \sum_{i \in \badcols} \tilde{ c }_i^{(t)}
        \left( 1+\frac{ \tilde{ c }_{i}}{n} -\frac 1k \right)
		\leq \fsmallsize\left( 1 - \frac 1{2k} \right),
	\end{eqnarray*}
	and then, by applying the Chernoff bound, we get that w.h.p. 
	\begin{eqnarray*}
        \sum_{i \in \badcols^{(t)}}   \tilde{ C }_i^{(t+1)} =
		\sum_{i \in \badcols^{(t)}}  C_i^{(t+1)} 
		+ \sum_i D_{i}^{(t+1)}  \leq \fsmallsize \left( 1-\frac1{4k} \right) + F 
		\leq \fsmallsize,
	\end{eqnarray*}
	concluding the proof.
\end{proof}

\begin{lemma}\label{le:stophp-adv}
Let $ \tilde{ \mathbf{c} }$ be any configuration such that
$|\bigc| = j$ and $\sum_{i \in \badcols} \tilde{ c }_i \leq \smallsize$.
Consider the stochastic process $\{ \tilde{ Y }_t\}_t$ defined as 
$ \tilde{ Y }_t = \left\lfloor\frac{n}{j}\right\rfloor 
    -  \tilde{ C }_{\mathbbm{m}}^{(t)}$ 
and define the stopping time 
\[  
        \tau = \inf \{ t \in \mathbb{N} 
        \,:\,  \tilde{ Y }_t \geqslant \sqrt{j n \log n} \,\vee\, 
        \big( \sum_{i \in \badcols} \tilde{ C }_i \geq \smallsize \big) \,\vee\,
        ( \bad^{(t-1)} \not \subseteq \bad^{(t)} )
        \}.
\]
Then, it holds that 
\[
    \Prob{\mathbf{c}}{ \tilde{ Y }_\tau > \alpha\sqrt{jn\log n}}\le\frac{1}{n}.
\]
\end{lemma}
\begin{proof}[Sketch of proof.]
    The proof of this Lemma follows from minor modifications of the proof 
    of Lemma \ref{le:stophp}. In particular, the argument is based on the 
    following observations: 

    \noindent {\bf 1.} The event defining the stopping time $\tau$ is in 
    this case 
    \[
        {\mathcal E}^{(t)} = \big(\tilde{ Y }_t \geq ( \sqrt{j n \log n} ) \,\vee\, 
            \big( \sum_{i \in \badcols} \tilde{ C }_i \geq \smallsize \big) \,\vee\,
            ( \bad^{(t-1)} \not \subseteq \bad^{(t)} )\big).
    \]
    The negated of this event is 
    \[
        \neg{\mathcal E}^{(t)} = \big(\tilde{ Y }_t \le ( \sqrt{j n 
        \log n} ) \,\wedge\, 
        \big( \sum_{i \in \badcols} \tilde{ C }_i \le \smallsize \big) 
        \,\wedge\,
        ( \bad^{(t-1)}  \subseteq \bad^{(t)} )\big),
    \]
    which implies the event $\left(\tilde{ Y }_t \le \sqrt{j n \log n} \right)$.
            
    \noindent {\bf 2.} Proceeding like in the proof of Lemma 
    \ref{le:stophp}, we can write an expression that is similar to 
    \eqref{eq:9}, with the generic conditioning event 
    \[
        C_{\mathbbm{m}}^{(s)} > \left\lfloor\frac{n}{j}\right\rfloor 
        - \sqrt{j n \log n}, 
    \]
    replaced by ${\neg\mathcal E}^{(s)}$. The conditioned event 
    \[
        C_{\ell}^{(t)} < \left\lfloor\frac{n}{j}\right\rfloor 
        - \alpha\sqrt{j n \log n},
    \]
    is instead replaced by the event 
    \[
        \left( C_{\ell}^{(t)} < \left\lfloor\frac{n}{j}\right\rfloor 
        - \alpha\sqrt{j n \log n}\right )\wedge {\mathcal E}^{(t)}. 
    \]
    Now, note that the event 
    \[
        C_{\ell}^{(t)} < \left\lfloor\frac{n}{j}\right\rfloor 
        - \alpha\sqrt{j n \log n},
    \] 
    again implies ${\mathcal E}^{(t)}$. Hence,  we can still write \eqref{eq:9},
    from which the proof requires some non-hard adaptations w.r.t. the case 
    without adversary.
\end{proof}

Since the adversary, at time $t$, may decide what to do based on the full 
history of the process up to time $t$, 
the stochastic process $\left\{ \mathbf{ \tilde {C} }^{(t)} \right\}_{t}$ 
may not be a Markov process anymore. Thus, we need a more general 
version of Lemma \ref{lemma:expectimewithdrift}. 

\begin{lemma}\label{lemma:expectimewithdrift-adv}
Let $\{X_t\}_t$ be a discrete time stochastic process 
with a finite state space $\Omega$, let $f_t :
\Omega ^t \rightarrow N$ be a function mapping histories of the process 
in non-negative integer numbers, and let $\{Y_t\}_t$ be the stochastic process 
over $N$ defined by $Y_t = f_t ( X_0, \dots, X_t )$. 
Let $m \in N$ be a ``target value'', 
let $A \subseteq \Omega$ be an arbirary subset of states, and let 
$$
\tau = \inf \{ t \in \mathbb{N} \,:\, Y_t \geqslant m \mbox{ or } X_t \notin A\}
$$
be the random variable indicating the first time $X_t$ exits from set $A$ or
$Y_t$ reaches or exceeds value $m$. 
Assume that, for every sequence of states $x_0, \dots, x_t \in A$ with
$f_t( x_0, \dots, x_t ) \leqslant m-1$, it holds that
\begin{enumerate}
\item (Positive drift). $\Expec{}{Y_{t+1} \,|\, X_0 = x_0, \dots, X_t = x_t } 
    \geq f_t(x_0, \dots, x_t) + \lambda$ for some $\lambda > 0$ 
\item (Bounded jumps).  $\Prob{x}{Y_{\tau} \geqslant \alpha m} \leqslant \alpha
    m /n$, for some $\alpha > 1$.
\end{enumerate}
Then, for every starting state $x \in A$, it holds that
$$
\Expec{x}{\tau} \leqslant 2 \alpha \frac{m}{\lambda}
$$ 
\end{lemma}
\begin{proof}
    The proof is a straight adaptation of the proof of Lemma 
    \ref{lemma:expectimewithdrift}, in which we take into account the
    full history of the process.

    Consider the stochastic process $Z_t = Y_t - \lambda t$.
    For any sequence of states $x_0, \dots, x_t \in A$ with $f_t(x_0, \dots, x_t) 
    \leq \border -1$ it holds that
    \begin{eqnarray*}
    \Expec{}{Z_{t+1} \,|\, X_0 = x_0, \dots, X_t = x_t} 
    & = & \Expec{}{Y_{t+1} \,|\, X_0 = x_0, \dots, X_t = x_t} - \lambda (t+1) \\
    & \geqslant & f_t(x_0, \dots, x_t) + \lambda - \lambda (t+1) \\
    & \geqslant & f_t(x_0, \dots, x_t)  - \lambda t
    \end{eqnarray*}
    where in the inequality we used Hypotheses \textit{1}. Thus, $Z_t$ is a
    \emph{submartingale} up to the stopping time $\tau$. 
    Moreover, since $|Y_t| \leq n$ then
    $
    |Z_{t+1} - Z_t| \leq n + \lambda
    $
    and, together with Hypotheses~\textit{1} this implies $\Expec{x}{\tau} <
    \infty$. Thus, we can apply \emph{Doob's Optional Stopping
    Theorem}~\cite{Doob53}. It follows that
    $\Expec{x}{Z_{\tau}} \geqslant \Expec{x}{Z_0} = f_0(x)$ and, since
    $\Expec{x}{Z_{\tau}} = \Expec{x}{Y_{\tau}} - \lambda \Expec{x}{\tau}$, we
    have that
    $$
    \Expec{x}{\tau} \leqslant \frac{\Expec{x}{Y_{\tau}} - f_0(x)}{\lambda} 
    \leqslant \frac{\Expec{x}{Y_{\tau}}}{\lambda}
    $$
    Finally, we get
    \begin{equation*}
    \Expec{0}{Y_{\tau}} = \sum_{j = 1}^{\lfloor \alpha \border  \rfloor} j
    \Prob{0}{Y_{\tau} =j} + \sum_{j = \lfloor \alpha \border  \rfloor + 1}^n j
    \Prob{0}{Y_{\tau} = j}
    \leq \left( \alpha \border  \right) + n \Prob{0}{Y_{\tau} > \alpha
    \border } \leq 2 \left( \alpha \border  \right)
    \end{equation*}
    where in the last inequality we used Hypothesis~\textit{2}.
\end{proof}

Now, we can exploit Lemma \ref{lemma:expectimewithdrift-adv}
to bound the time required by the symmetry-breaking stage: We
show that, from any configuration with $j$ big opinions that satisfies
\hhyp, within
$\mathcal{O}\left(j^2 \log^{1/2} n\right)$ rounds there exists a opinion
supported by at most $n/j - \sqrt{j n \log n}$ nodes with probability at least
$1/2$.
\begin{lemma}[Symmetry-breaking stage]
	\label{le:breaking-adv}
    Let $ \tilde{ \mathbf{c} }$ be any configuration such that 
    $|\bigc|  = j$ and $\sum_{i \in \badcols} \tilde{ c }_i \leq \smallsize$. 
    Within $t = \mathcal{O}\left(j^2 \log^{1/2} n\right)$ rounds, 
    with probability at least $1/2$ it holds that 
    \[
        {|\bigc|  = j \mbox{, } \sum_{i \in \badcols} \tilde{ C }_i 
        \leq \smallsize \mbox{ and } \exists i \in \bigc^{(t)} 
        \mbox{ such that }  
        \tilde{ C }_i^{(t)} \leqslant n/j - \sqrt{j n \log n}}
    \]
\end{lemma}
\begin{proof}
We proceeds by adapting the proof of Lemma \ref{lemma:breaking}.
    Let $ \tilde{ \col }^{(0)} =  \tilde{ \mathbf{c} }$ 
    be the initial configuration.
    Let us consider the stochastic process $\{ \tilde{ Y }_t\}_{t \geq 0}$ 
    defined as 
    \[ 
         \tilde{ Y }_t = \left\lfloor\frac{n}{j}\right\rfloor 
         - \tilde{ C }_{\mathbbm{m}}^{(t)} 
    \]
    where 
    $ \tilde{ C }_{\mathbbm{m}}^{(t)} 
        = \min\{ \tilde{ C }_i^{(t)} : i \in \bigc^{(t)} \}$.
    We are interested in the time step
    \[ 
        \tau = \inf \{ t \in \mathbb{N} \,:\,  
        \tilde{ Y }_t \geq ( \sqrt{j n \log n} ) \,\vee\, 
        \big( \sum_{i \in \badcols} \tilde{ C }_i \geq \smallsize \big) \,\vee\,
        ( \bad^{(t-1)} \not \subseteq \bad^{(t)} )
        \} 
    \]
    Now we show that $\{ \tilde{ Y }_t\}_t$ satisfies
    the Hypotheses \textit{1} and \textit{2} of
    Lemma~\ref{lemma:expectimewithdrift} 
    with 
    $A = \big( \sum_{i \in \badcols} \tilde{ C }_i 
        \leq \smallsize \big) \,\vee\,
        ( \bad^{(t-1)} \subseteq \bad^{(t)} )$ 
    and $\lambda = \varepsilon
    \sqrt{n}/j^{3/2}$, for a suitable constant $\varepsilon > \alpha $.

    \smallskip\noindent \textit{1.} Let 
    ${ \tilde{ \mathbf{c} }} $ 
    be any configuration such that 
    $\tilde{c}_{\mathbbm{m}} > n/j - \sqrt{j n \log n}$. 
    Now we prove that 
    \begin{equation}\label{eq:cmindrift-adv}
        \Expec{}{ \tilde{ C }_\mathbbm{m}^{(t+1)} \,|\, 
        \tilde{ \mathbf{C} }^{(t)} =
        \tilde{\mathbf{c}}} \leq  \tilde{ c }_{\mathbbm{m}} - \varepsilon
        \frac{\sqrt{n}}{j^{3/2}} 
    \end{equation}

    \smallskip\noindent 
    \underline{Case $\tilde{c}_{\mathbbm{m}} > n/j - 2
    \varepsilon \sqrt{n/j}$:} Observe that, in this case, random variables
    $\left\{ { C }_i^{t+1} \,:\, i \in \bigc \right\}$ 
    have  standard deviation is $\Omega (\sqrt{n/j} )$. 
    Moreover they are binomial and negatively
    associated.  Hence,  by choosing $\varepsilon$ small enough, from the
    Central Limit Theorem we have that 
    \[ 
        \Prob{}{i \in \bigc \mbox{ exists such
        that } C_i^{(t+1)} \leqslant \frac{n}{j} - 6 \varepsilon \cdot
        \sqrt{\frac{n}{j}}} \geqslant 1/2
    \]
    We thus get 
    \begin{eqnarray}
        \label{eq:cmincase2-adv} 
        \Expec{}{ \tilde{ C }_\mathbbm{m}^{(t+1)} \,|\,
        \tilde{ \mathbf{C} }^{(t)} = \tilde{\mathbf{c}}} 
        & \leq & \frac{1}{2} \left(
        \frac{n}{j} - 6 \varepsilon \cdot \sqrt{\frac{n}{j}} \right) +
        \frac{1}{2} \cdot \frac{n}{j} + \fadv \nonumber \\ 
        & = & \frac{n}{j} - 2 \varepsilon \sqrt{\frac{n}{j}} + \fadv
        \leq  \tilde{ c }_{\mathbbm{m}} - \varepsilon
        \sqrt{\frac{n}{j}} \leq \tilde{ c }_{\mathbbm{m}} - \varepsilon
        \frac{\sqrt{n}}{j^{3/2}} 
    \end{eqnarray}

    \smallskip\noindent \underline{Case $\tilde{c}_{\mathbbm{m}} \leqslant n/j -
    2 \varepsilon \sqrt{n/j}$:} Equation \eqref{eq:cmindrift-adv} easily follows
    from Lemma \ref{lem:average-adv}. Indeed, let $i \in \bigc$ be a opinion such that
    $\hat{c}_i = \hat{c}_{\mathbbm{m}}$, then 
    \begin{eqnarray*}
        \Expec{}{ \tilde{ C }_\mathbbm{m}^{(t+1)} \,|\, 
        \tilde{ \mathbf{C} }^{(t)} =
        \tilde{\mathbf{c}}} & \leq & \Expec{}{ \tilde{ C }_i^{(t+1)} \,|\,
        \tilde{ \mathbf{C} }^{(t)} 
        = \tilde{\mathbf{c}}} \leq \tilde{c}_i \left( 1 + \frac{\tilde{c}_i 
        + \smallbias}{n} 
        - \frac{1}{j} \right) \\ 
        & \leq & \tilde{c}_i \left( 1 - \frac{2\varepsilon}{\sqrt{nj}} 
        + \frac { \alpha }{ \sqrt{kn} }
        \right) \leq \tilde{c}_i - \frac{\varepsilon\sqrt{n}}{j^{3/2}} 
        = \tilde{c}_{\mathbbm{m}} - \varepsilon \frac{\sqrt{n}}{j^{3/2}} 
    \end{eqnarray*} 
    where  we used the case's condition   and   
    $\tilde{c}_i = \tilde{c}_{\mathbbm{m}} \geqslant n/(2j)$.

    \smallskip\noindent \textit{2.}
    Since random variables 
    $\left\{ \tilde{ C }_i^{(t)} \,:\, i \in \bigc^{(t)} \right\}$ 
    are binomial conditional on the configuration at round $t-1$, 
    from the Chernoff bound it follows that 
    \begin{equation}
    \label{eq:hypjumpnew-adv}
    \Prob{ \tilde{ \mathbf{c} }}{ \tilde{ Y }_{\tau} 
    \geq \alpha \sqrt{j n \log n}} \leq \frac{1}{n}, \
    \mbox{ for  some    constant } \ \alpha > 1 
    \end{equation}
    See Lemma \ref{le:stophp-adv} for the formal statement of the last fact.
 
    From \eqref{eq:cmindrift-adv} and \eqref{eq:hypjumpnew-adv} we
    have that $\{ \tilde{ Y }_t\}_t$ satisfies the hypotheses
    of~Lemma~\ref{lemma:expectimewithdrift-adv} with $m = \sqrt{j n \log n}$, 
    $\lambda = \varepsilon \sqrt{n}/j^{3/2}$ and 
    $A = \big( \sum_{i \in \badcols} \tilde{ C }_i 
        \leq \smallsize \big) \,\vee\,
        ( \bad^{(t-1)} \subseteq \bad^{(t)} )$. 
    Moreover, by iteratively applying Lemma \ref{le:hyp-H}, we have
    that, for any $t = \bigO( n^2 )$, it holds w.h.p. that 
    $\big( \sum_{i \in \badcols} \tilde{ C }_i^{(t)} 
        \leq \smallsize \big) \,\vee\,
        ( \bad^{(t-1)} \subseteq \bad^{(t)} )$.
    Thus, from Markov's inequality, for $t = 2 j^2 \sqrt{\log n}$,
    we have that 
    \begin{multline*}
        \mathbf{P}_{ \tilde{ \mathbf{c} }}\big( \forall i \in \bigc,\,
        \big( C_i^{(t)} \leq n/j - \sqrt{j n \log n} \big) \,\wedge\, 
            \big( \sum_{i \in \badcols} \tilde{ C }_i^{(t)} 
            \leq \smallsize \big) \,\wedge\,
            ( \bad^{(0)} \subseteq \bad^{(t)} ) \big) \\
        \geq \Prob{ \tilde{ \mathbf{c} }}{\hat \tau \leq 2 j^2 \sqrt{\log n}} 
            \geq \frac{1}{3}
    \end{multline*}
    where $ \hat\tau = \inf \{ t \in \mathbb{N} \,:\, 
            \tilde{ Y }_t \geq \sqrt{j n \log n} \}$.
	
\end{proof}

Once the smallest opinion goes below the average size
of the big opinions by a certain small amount, 
we can prove that the process push it in the set
of small opinions w.h.p.

\begin{lemma}[Dropping stage]
    Assume that, at round $\timestart$, $\tilde{ \mathbf{c}}^{(\timestart)}$ 
    is such that $\sum_{i \in \badcols}c_i \leq \smallsize$,
    $|\bigc^{(\timestart)}| = j$, and an $i\in \bigc^{(\timestart)}$ exists
    such that $ \smallsize \leq c_i^{(\timestart)} \leq n/j - \sqrt{kn\log n}
    $. Then, w.h.p.,  a round $\timeend =\timestart + {O}(k\log n) $ exists such
    that $\sum_{i \in \badcols} \tilde{ C }_i^{(\timeend)} \leq \smallsize$, 
    $i\in \bad^{(\timeend)}$ and $|\bigc^{(\timeend)}| \leq j-1 $.
	\label{le:dying-adv}
\end{lemma}
\begin{proof}
    By iteratively applying Lemma \ref{le:hyp-H}, we have that, w.h.p., 
    for each $t \in \left\{ t', \dots, t''-1 \right\}$ it holds
	$\sum_{i \in \badcols} \tilde{ C }_i^{(t)} \leq \smallsize$ and
    $i\in \bad^{(t)}$.

    To prove that $|\bigc^{(\timeend)}| \leq j-1 $, we first prove that, for each round $t \in \left\{
    \timestart + 1, \dots, \timeend \right\}$,
    w.h.p. $ \tilde{ C }_i^{(t)}$ decreases by a certain extent that depends
    on $ \tilde{ c }_i^{(t-1)}$, regardless of what the adversary does.

    Let $ \dieadvfactor = ( {1} / {j} - ( \tilde{ c }^{(t-1)}_{i} 
    + \smallbias) / {n})$
	If we are in a configuration satisfying the hypotheses of
	the lemma, we have
	\begin{eqnarray*}
	\Prob{}{C_{i}^{(t)}> \tilde{ c }^{(t-1)}_{i}
        \left(1-\frac{\dieadvfactor}{2} \right)} 
        = \Prob{}{C_{i}^{(t)}> \tilde{ c }^{(t-1)}_{i}
        \left(1-\dieadvfactor(1 + \delta) \right)}, 	
	\end{eqnarray*}
	where $ \delta = {\frac{1}{2}\dieadvfactor} / ({1-\dieadvfactor})$.
	Thus, using Lemma \ref{lem:average-adv} and applying the Chernoff bound we have
	\begin{eqnarray}
	\Prob{}{C_{i}^{(t)}> \tilde{ c }^{(t-1)}_{i}
	\left(1-\frac{\dieadvfactor}{2}\right)} 
    \leq \exp\left\{ -\frac{\delta^2}{3} \dieadvfactor
        \tilde{ c }^{(t-1)}_{i}\right\} 
	< \exp\left\{ - \frac{1}{3}\left(\frac{1}{2}\dieadvfactor\right)^2 
    \tilde{ c }^{(t-1)}_{i}\right\} 
	= n^{-\Theta\left(1\right)},
	\label{eq:hp-adv-4}
	\end{eqnarray}
	where the second inequality follows from the definition of $\delta$ and
	the fact that its denominator is smaller than $1$, and the
    equality in \eqref{eq:hp-adv-4} follows by minimizing 
    $\dieadvfactor^2 \tilde{ c }^{(t-1)}_{i}$ for 
    $\smallsize \leq c_i^{(\timestart)} \leq n/j - \sqrt{kn\log n}$.
    
    It follows that, w.h.p.
    \begin{equation}
        \label{eq:hp-adv-5}
            \tilde{ C }_i^{(t)} = { C }_i^{(t)} + D_i^{(t)} \leq  
            \tilde{ c }^{(t-1)}_{i} 
            \left(1-\frac{\dieadvfactor}{2}\right) 
            + F \leq \tilde{ c }^{(t-1)}_{i} 
    \end{equation}
	
	Thus, w.h.p., we can iteratively apply \eqref{eq:hp-adv-5} until 
	$ \tilde{ c }^{(t-1)}_{i} \leq \smallsize$.
    We next prove that this happens within $\mathcal{O}\left(k\log 
	n\right)$ rounds, w.h.p. 
	Interestingly, showing that, within $\mathcal{O}\left(k\log 
	n\right)$ rounds, $C_i$ decreases to a costant fraction of its 
	value at the beginning of the dropping stage does not seem obvious. 
	For this reason, we consider the evolution of the displacement $\frac{n}{j}-C_{i}$, which seems 
	analytically more tractable. 
	To this purpose, note that \eqref{eq:hp-adv-4} implies that, w.h.p.
    \begin{eqnarray}
	\frac{n}{j}-C_{i}^{(t)} 
    & \geq& \frac{n}{j} - c^{(t-1)}_{i} 
    + \frac{c^{(t-1)}_{i}}{2}\left(\frac{1}{j}-\frac{c^{(t-1)}_{i} + \smallbias}{n}\right) \nonumber\\
    &=&  \frac{n}{j} - c^{(t-1)}_{i} 
    + \frac{c^{(t-1)}_{i}}{2}\left(\frac{1}{j}-\frac{c^{(t-1)}_{i}}{n} 
    \right) \left( 1 - \frac{ \smallbias}{ 1/j - c^{(t-1)}_{i}/n } \right) \nonumber \\
    &=& \frac{n}{j} - c^{(t-1)}_{i} 
    + \frac{c^{(t-1)}_{i}}{2}\left(\frac{1}{j}-\frac{c^{(t-1)}_{i}}{n} 
    \left( 1 + \frac{ \alpha }{ \log n } \right) \right) \nonumber \\
    &=& \left(\frac{n}{j} 
    - c^{(t-1)}_{i}\right) \left(1 + \constone \frac{c^{(t-1)}_{i}}{2n}\right)
    \label{eq:dying_whp}
    \end{eqnarray}
    for some constant $ \constone > 0 $, where in the first equality of 
    \eqref{eq:dying_whp} we have used that $ n / j - c^{(t-1)}_{i} \geq 
    \sqrt{ k n \log n }$. 
	
	We can now conclude the proof of Lemma~\ref{le:dying-adv}.
	We first prove that $C_{i}\leq n/\left(2j\right)$ within 
	$\mathcal{O}\left(k\log n\right)$ steps, w.h.p. To this purpose, note that 
	$\frac{n}{j}-c_{i}\ge \sqrt{kn\log n}$ at the beginning of the droppign stage 
	from the hypotheses. 
	Furthermore, for some positive constants $\consttwo$ and $\constthree$, 
    as long as $C_{i} \geq \constthree  {n} / {j} $ it holds 
    $1 + \constone {c_{i}} / {n} \geq 1 + {\consttwo} / {j}$. 
    Hence, after $\mathcal{O}\left(k\log n\right)$ steps, 
	w.h.p. we have $ \frac{n}{j}-c_{i} \geq (1-\constthree)\frac{n}{j}$, 
    which in turn implies $c_{i} \leq \constthree n / j$.
	Once $c_{i} \leq \constthree n / j $, using again \eqref{eq:dying_whp} we 
	have that $C_i$ decreases by a factor $1 - \Omega(1/j)$ 
	in every round w.h.p. By standard concentration arguments we obtain 
    that eventually $c_{i} \leq \smallsize$ within 
	$\mathcal{O}\left(k\log n\right)$ more steps, w.h.p.
\end{proof}

\end{document}